\documentclass[a4paper,USenglish]{lipics-v2018-derivative}
\hideLIPIcs
\usepackage{amsmath}
\usepackage{tikz}
\usetikzlibrary{arrows.meta, decorations.pathreplacing, shadows, patterns}
\usepackage{colortbl, multirow}
\usepackage{hyperref}
\usepackage{algorithm,algpseudocode}

\theoremstyle{plain}
\newtheorem{problem}{Problem}
\newtheorem{claim}{Claim}
\newtheorem{conjecture}{Conjecture}
\newtheorem{fact}{Fact}

\newcommand{\ED}{\mathsf{ED}}
\newcommand{\Ham}{\mathsf{HAM}}
\newcommand{\Oh}{\mathcal{O}}
\newcommand{\prob}[1]{\mathop{\mathrm{Pr} \left[#1\right]}}
\newcommand{\eps}{\varepsilon}
\newcommand{\blockdist}{2}
\newcommand{\Q}{\mathcal{Q}}
\newcommand{\D}{\mathcal{D}}

\title{Lower bounds for text indexing with mismatches and differences}

\author{Vincent Cohen-Addad}{Sorbonne Universit\'{e}, UPMC Univ Paris 06, CNRS, LIP6}{vcohenad@gmail.com}{}{Partially supported by CNRS JCJC project EASSS and ANR JCJC project FOCAL.}

\author{Laurent Feuilloley}{IRIF, CNRS and University Paris Diderot}{feuilloley@irif.fr}{}{Partially supported by ANR project DESCARTES.}

\author{Tatiana Starikovskaya}{DI/ENS, PSL Research University}{tat.starikovskaya@gmail.com}{}{Partially supported by CNRS JCJC project EASSS.}

\authorrunning{V. Cohen-Addad, L. Feuilloley, T. Starikovskaya}

\subjclass{Pattern matching, Cell-probe models and lower bounds}


\begin{document}

\date{\empty}
\maketitle
\begin{abstract}
In this paper we study lower bounds for the fundamental problem of text indexing with mismatches and differences. In this problem we are given a long string of length $n$, the ``text'', and the task is to preprocess it into a data structure such that given a query string $Q$, one can quickly identify substrings that are within Hamming or edit distance at most $k$ from $Q$. This problem is at the core of various problems arising in biology and text processing.

While exact text indexing allows linear-size data structures with linear query time, text indexing with $k$ mismatches (or $k$ differences) seems to be much harder: All known data structures have exponential dependency on $k$ either in the space, or in the time bound. We provide conditional and pointer-machine lower bounds that make a step toward explaining this phenomenon.
  
We start by demonstrating lower bounds for $k = \Theta(\log n)$. We show that assuming the Strong Exponential Time Hypothesis, any data structure for text indexing that can be constructed in polynomial time cannot have $\Oh(n^{1-\delta})$ query time, for any $\delta>0$. This bound also extends to the setting where we only ask for $(1+\eps)$-approximate solutions for text indexing. 

However, in many applications the value of $k$ is rather small, and one might hope that for small~$k$ we can develop more efficient solutions. We show that this would require a radically new approach as using the current methods one cannot avoid exponential dependency on $k$ either in the space, or in the time bound for all even $\frac{8}{\sqrt{3}} \sqrt{\log n} \le k = o(\log n)$. Our lower bounds also apply to the dictionary look-up problem, where instead of a text one is given a set of strings.
\end{abstract}

\section{Introduction}
\label{sec:introduction}
Text indexing is the task of preprocessing a ``text'' (i.e., a long string) into a data structure such that given a query string $Q$, one can quickly identify all substrings of the text that are equal to~$Q$. Text indexing is an essential primitive in the context where a large database has been collected, and new elements have to be compared to this database. As such it has applications ranging from biology to signal processing and text retrieval (see \cite{Navarro01}, where these applications are described for the similar problem of approximate string matching).

For most applications, searching for exact matches only is not sufficient. Indeed, if the database or the query contains noisy data, then looking only for exact matches will make us miss most of the matches that we would like to report. Also, even if the data is clean, how to efficiently identify sequences that are only \emph{similar} to the query is an important question. 
For example, given a gene
one would like to know if some variant of it is present in a given genome. 
Then one has to consider a variant of text indexing that asks to quickly identify substrings of the text that are similar to  the query. The importance of this task is witnessed by the extensive use of the program BLAST that serves exactly this purpose for biologists (the original paper on BLAST \cite{AGM+90} is cited 50,000+ times, see also~\cite{Myers13}). 

There are various ways of defining similarity between strings. Arguably, the two most common similarity measures are the Hamming and the edit distances.
The Hamming distance is the minimum number of symbols one has to change to go from one string to the other when the two strings are aligned.
The edit distance is similar but insertions and deletions are also allowed. 
In other words, Hamming distance measures the number of \emph{mismatches} only, and edit distance can also cope with \emph{insertions} and \emph{deletions}. Since the seminal paper of  Cole, Gottlieb, and  Lewenstein~\cite{kerratatree}, the literature has seen a number of works dedicated to upper bounds for text indexing with mismatches and differences, however very few lower bounds are known and so the complexity of the problem remains to be established. 

The resources considered in this context are the space used by the data structure and the time it takes to answer a query. 
Despite significant progress, no text index has crossed what Navarro
conjectured in 2010 to be a fundamental time-space barrier~\cite{ConjectureNavarro}: Namely, for $k$ being the maximum allowed number of mismatches or differences, either the size of the index or the query time must depend on $k$ exponentially.

\paragraph*{Our results.}
In this paper, we prove this conjecture by giving lower bounds, showing that indeed getting over this bound is beyond the
current techniques.
The lower bounds for data structures have to refer to a precise model, and we focus on the classic RAM model and the pointer machine model. 
In the RAM model, we give lower bounds conditional on the now classic
Strong Exponential Time Hypothesis (SETH), and in the pointer machine model, we provide unconditional lower bounds. 
The pointer machine model is more restricted than the RAM model, and is arguably less popular, nevertheless it is very relevant for our purposes as all known solutions for text indexing with mismatches and differences are pointer-machine data structures. We provide lower bounds for both the Hamming and the edit distances.
We give a more detailed overview of the results and techniques we obtain in Section~\ref{sec:ourresults},
see also Fig.~\ref{fig:SETH} and~\ref{fig:pointer-machine}.

Finally, we highlight two additional features of the paper. First, our lower bounds also apply to the \emph{dictionary look-up} problem, where instead of a text one is given a set of strings, and which is also a very classic framework in practice. Second, to derive the lower bounds for the edit distance we exploit a construction that relates the Hamming distance and the edit distance. 
This step that we call \emph{stoppers transform} could be of independent interest as it allows to construct hard instances for the edit distance from hard instances for the Hamming distance, and the latter is much easier to analyse.

\begin{figure*}[ht!]
\begin{tikzpicture}[scale =0.65]
\node[rectangle,align=center]  (a) at (0,3) {\textbf{Bichromatic closest pair}\\ \textbf{with mismatches} \\ $\Omega(n^{2-\delta})$ time~(\cite{AW:2015,R:2018})}; 
\node[rectangle,align=center]  (b) at (0,0) {\textbf{Bichromatic closest pair}\\ \textbf{with differences}\\ $\Omega(n^{2-\delta})$ time \\ (\cite{R:2018}, Corollary~\ref{cor:BCD})};

\node[rectangle,align=center]  (c) at (7.5,3) {\textbf{Dictionary look-up}\\ \textbf{with $k$ mismatches} \\ $\Omega(n^{1-\delta})$ query time~(\cite{R:2018})}; 
\node[rectangle,align=center]  (d) at (7.5,-0) {\textbf{Dictionary look-up} \\ \textbf{with $k$ differences}\\  $\Omega(n^{1-\delta})$ query time \\(\cite{R:2018}, Corollary~\ref{cor:DLD})};

\node[rectangle,align=center]  (e) at (15.5,3) {\textbf{Text indexing} \\ \textbf{with $k$ mismatches}\\ $\Omega(n^{1-\delta})$ query time \\ (Corollaries~\ref{cor:SETH-text-indexing-HAM} and~\ref{cor:SETH-approx-text-indexing-HAM})}; 
\node[rectangle,align=center]  (f) at (15.5,-0) {\textbf{Text indexing} \\ \textbf{with $k$ differences}\\ $\Omega(n^{1-\delta})$ query time \\ (Corollaries~\ref{cor:SETH-text-indexing-ED} and~\ref{cor:SETH-approx-text-indexing-ED})}; 

\draw[->] (a) edge node[pos=0.5,right] {\tiny{Stoppers transform}, \cite{R:2018}} (b) ;
\draw[->] (a) edge node[pos=0.5,above] {} (c);
\draw[->] (b) edge node[pos=0.5,above] {} (d);

\draw[->] (c) edge node[pos=0.5,above] {\tiny{Theorem~\ref{th:dictionary-textindexing-HAM}}} (e) ;
\draw[->] (d) edge node[pos=0.5,above] {\tiny{Theorem~\ref{th:dictionary-textindexing-ED}}} (f) ;
\end{tikzpicture}
\caption{Lower bounds conditional on SETH. These bounds hold for all data structures that can be constructed in polynomial time, both in the exact and $(1+\eps)$-approximate settings, and~${k = \Theta(\log n)}$. We show a simpler proof of the lower bound for dictionary look-up with differences via the stoppers transform and give new lower bounds for text indexing. }
\label{fig:SETH}
\end{figure*}
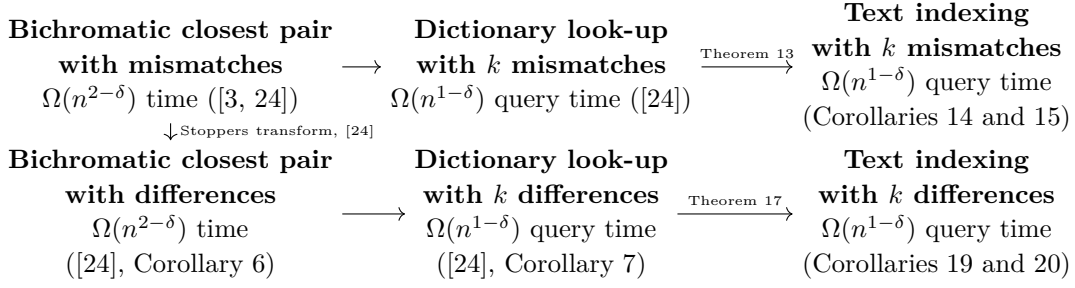

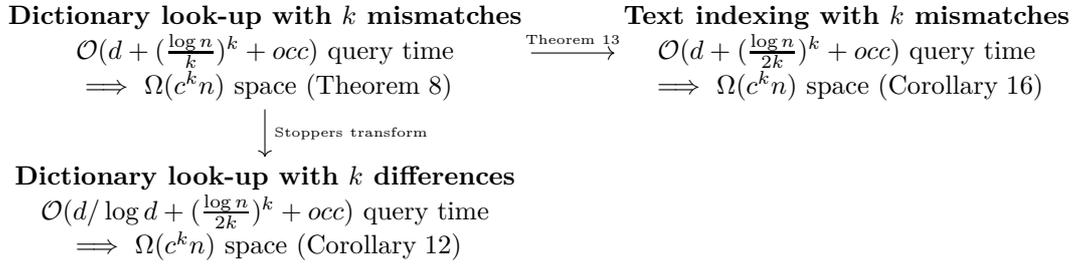
\begin{figure*}
\begin{tikzpicture}[scale =0.85]
  \node[align=center]  (a) at (9,2.5) {\textbf{Text indexing with $k$ mismatches}\\
   $\Oh(d + (\frac{\log n}{2k})^k + occ)$ query time\\ $\implies$ $\Omega(c^k n)$ space (Corollary~\ref{cor:pointer-machine-text-indexing-HAM})}; 

\node[align=center]  (c) at (0,2.5) {\textbf{Dictionary look-up with $k$ mismatches}\\
  $\Oh(d + ( \frac{\log n}{k})^k + occ)$ query time\\ $\implies$ $\Omega (c^{k} n)$ space (Theorem~\ref{th:pointer_lower_HD_dict})}; 
\node[align=center]  (d) at (0,0) {\textbf{Dictionary look-up with $k$ differences} \\
  $\Oh(d / \log d + (\frac{\log n}{2k})^k + occ)$ query time\\
  $\implies$ $\Omega(c^k n)$ space (Corollary~\ref{cor:pointer-machine-dictionary-ED})}; 

\draw[->] (c) edge node[pos=0.5,right] {\tiny{Stoppers transform}} (d) ;
\draw[->] (c) edge node[pos=0.5,above] {\tiny{Theorem~\ref{th:dictionary-textindexing-HAM}}} (a);
\end{tikzpicture}
\caption{Summary of new pointer-machine lower bounds for dictionary look-up and text indexing. Here $d$ is the length of the query string (pattern) and $c > 1$ is some constant. The bounds hold for all even $k$ such that $\frac{8}{\sqrt{3}} \sqrt{\log n} \le k = o(\log n)$. }
\label{fig:pointer-machine}
\end{figure*}

\subsection{Related work and background.}
Many of the problems below have been considered both in the field of algorithms on strings and in the field of computational geometry, and sometimes they are known under different names. We try, where possible, to provide alternative names. 

\subsubsection{Dictionary look-ups.}
The famous problem of \emph{dictionary look-up with mismatches or differences} was introduced by Minsky and Papert in 1968~\cite{perceptrons}. The problem is also known under the name of \emph{$k$-neighbour}.

\begin{problem} \textbf{Dictionary look-up with $k$ misma\-tches (differences)}\\
\emph{Input:} An alphabet $\Sigma$, a set of strings in $\Sigma^d$ of total length $n$ (dictionary).\\
\emph{Output:} A data structure that given a string $Q \in \Sigma^d$, outputs all strings in the dictionary that are at Hamming (edit) distance $\le k$ from $Q$.
\end{problem}

In this work, we focus on data structures with worst-case space and time guarantees, and therefore we do not consider results where either the query time bound or the space bound hold on average only. We also ignore the case of $k = 1$, which has been extensively studied in the literature, as in this case the problem has a very special structure: for example, if two strings $S$ and $Q$ are at Hamming distance at most one, there must exist an integer~$i$ such that $Q[1,i-1]Q[i+1,|Q|] = S[1,i-1]S[i+1,|S|]$, and one has a similar property for the edit distance. 

For $k > 1$, the seminal paper of Cole, Gottlieb, and Lewenstein~\cite{kerratatree} presented a data structure for dictionary look-up with mismatches called \emph{$k$-errata tree} that uses $\Oh(n \log^k n)$ space and has query time $\Oh(d + \frac{1}{k!} (c \log n)^k \log \log n + occ)$ for some constant $c > 1$. They also showed that a variant of this data structure can be used to solve dictionary look-up with differences in $\Oh(n \log^k n)$ space and query time $\Oh(d + \frac{1}{k!} (c \log n)^k \log \log n + 3^k occ)$. 
We would also like to briefly mention another branch of research dedicated to approximate solutions to the problem of dictionary look-up with~$k$ mismatches, and in particular locality-sensitive hashing (LSH) techniques (see~\cite{ANNsurvey} for a survey). The LSH-based dictionary look-up data structures guarantee that if the dictionary contains a string within Hamming distance $\le k$ from the query string $Q$, then the data structure will return a string within Hamming distance $\le (1+\eps) k$ from $Q$ with constant probability. For the edit distance, no efficient solution with a constant approximation factor exists.

In~\cite{BR:2002,BOR:1999} it was shown that for some constant $c > 0$, $d \in w(\log n \cap n^{o(1)})$, and $k = d/2 - c \sqrt{d \log n}$, any randomised two-sided error cell-probe algorithm for dictionary look-ups with $k$ mismatches that is restricted to use $\mathrm{poly} (n,d)$ cells of size $\mathrm{poly}(\log n, d)$ each, must probe at least $\Omega(d / \log n)$ cells. Since the cell-probe model is stronger than the RAM model, this lower bound implies $\Omega(d/\log n)$ query time lower bound for the classic RAM data structures. We note, however, that the lower bound is not that high while the value of~$k$ is rather large, $k = w(\log n \cap n^{o(1)})$.

In the related problem of \emph{dictionary look-up with $k$ don't cares}, the query strings may contain up to $k$ don't care symbols, that is, special symbols that match any symbol of the alphabet, and the task is to retrieve all dictionary strings that match the query string. Essentially, it is a parameterized variant of the \emph{partial match} problem (see~\cite{P:2011} and references therein). The structure of this problem is very similar to that of dictionary look-up with mismatches and differences, as each don't care symbol in the query string gives $|\Sigma|$ possibilities for the corresponding symbol in the dictionary strings. On the other hand, it is simpler, as the positions of the don't care symbols are fixed. However, even for this simpler problem all known solutions have an exponential dependency on $k$, either in the space complexity or in the query time (see~\cite{LMRT:2014} and references therein). Afshani and Nielsen~\cite{AN2016} showed that in general one cannot avoid this. In more detail, they showed that for $3 \sqrt{\log n} \le k = o(\log n)$, any pointer-machine data structure with query time $\Oh(2^{k/2} + d + occ)$ requires $\Omega(n^{1+\Theta(1/\log k)})$ space, even for the binary alphabet.

\subsubsection{Text indexing.}
\begin{problem}\textbf{Text indexing with $k$ mismatches (differences)}\quad\\
\emph{Input:} An alphabet $\Sigma$, a string $T \in \Sigma^n$ (referred to as \emph{text}), an integer $k$.\\
\emph{Output:} A data structure (referred to as \emph{text index}) that maintains the following queries: Given a string $Q \in \Sigma^d$ (referred to as \emph{pattern}), output a substring / all substrings of $T$ that are within Hamming (edit) distance $k$ from $Q$.
\end{problem}

\begin{figure*}
\begin{center}
\begin{tabular}{|l|l|l|}
\hline
\rowcolor{gray!10}
\textbf{Space} & \textbf{Query time} & \\
\hline
\hline
$\Oh(n (\alpha \log \alpha \log n)^k)$ words &  $\Oh(d + (\log_\alpha n)^k\log \log n + occ)$& \cite{Tsur2010FastIF} \\
$\Oh(n \log^k n)$ words & $\Oh(d+\frac{1}{k!} (c \log n)^k \log \log n + occ)$ & \cite{kerratatree} \\
\hline
$\Oh(n \log^{k-1} n)$ words & $\Oh(d+\frac{1}{k!} (c \log n)^k \log \log n + occ)$ & \cite{CLSTW:2006-cpm} \\
\hline
\multirow{5}{*}{$\Oh(n)$ words}   & $\Oh(d^2 \min\{n, |\Sigma|^k d^{k+1}\} +occ)$ & \cite{U:1993}\\
													  & $\Oh(d \min\{n, |\Sigma|^k d^{k+1}\} \log \min\{n, |\Sigma|^k d^{k+1}\} +occ)$ & \cite{U:1993}\\
													  & $\Oh(\min\{n, |\Sigma|^k d^{k+2}\} +occ)$ & \cite{CGU:1995}\\
													  & $\Oh(\min\{ ( |\Sigma| d)^k \log n +occ)$ & \cite{HHLS:2006} \\
 													  & $\Oh(d+(c \log n)^{k(k+1)} \log \log n + occ)$ & \cite{CLSTW:2006-cpm} \\
													  & $\Oh(|\Sigma|^k d^{k-1} \log n \log \log n + occ)$ & \cite{CLSTW:2010}\\
\hline
$\Oh(n \sqrt{\log n})$ bits & $\Oh((|\Sigma| d)^k \log \log n + occ)$ & \cite{LSW:2008} \\
\hline
\multirow{3}{*}{$\Oh(n)$ bits}  & $\Oh((|\Sigma| d)^k \log^2 n + occ \log \log n)$  & \cite{HHLS:2006} \\
& $\Oh(((|\Sigma| d)^k \log \log n + occ) \log^{\delta} n)$  & \cite{LSW:2008} \\
& $\Oh((d + (c \log n)^{k^2+2k}\log\log n +occ) \log^{\delta} n)$ & \cite{CLSTW:2006-cpm}\\
\hline
\end{tabular}
\end{center}
\caption{Upper bounds for text indexing with $k $ mismatches. Here $\alpha$ is any integer in $[2,n/2]$, $c > 1$, $\delta > 0$ are constants, and $occ$ is either one, or the total number of the substrings that we output depending on the version of the problem. For text indexing with differences $k$, the term $occ$ is replaced with $3^k occ$.}
\label{table:textindexing}
\end{figure*}
Cole, Gottlieb, and Lewenstein~\cite{kerratatree} also extended $k$-errata trees to text indexing. 
Subsequent work has mainly focused on improving the space requirements~\cite{CLSTW:2010,CLSTW:2006-cpm,HHLS:2006,LSW:2008}. One work that stands apart is~\cite{Tsur2010FastIF}, where the author showed an index with improved query time by sacrificing the space. For the summary of results for text indexing, see Table~\ref{table:textindexing}.

Andoni and Indyk~\cite{AI:2006} introduced an approximate version of text indexing and showed that it can be solved via the LSH technique. 

\begin{problem}\textbf{$(1+\eps)$-approximate text indexing\;\; with $k$ mismatches}\\
\emph{Input:} An alphabet $\Sigma$, a string (text) $T \in \Sigma^n$, an integer $k$, a constant $\eps > 0$.\\
\emph{Output:} A text index that maintains the following queries: Given a pattern $Q \in \Sigma^d$ such that there is a substring of $T$ within Hamming distance at most $k$ from $Q$, output a substring of $T$ that is within Hamming distance at most $(1+\eps)k$ from $Q$.
\end{problem}

Andoni and Indyk showed that the problem can be solved in $\tilde{\Oh}(n^{1+1/(1+\eps)})$ space and $\tilde{\Oh}(n^{1/(1+\eps)} + d n^{o(1)})$ query time, where $\tilde{\Oh}$ hides factors that are polylogarithmic in $n$. We note that the statement of the problem is a bit unusual for string processing, however, the index returns meaningful answers, and is more efficient than the known exact solutions for text indexing with mismatches when $k$ is large. 

\subsubsection{Bichromatic closest pair.}
\begin{problem}\textbf{Bichromatic closest pair with mismatches (differences)}\quad\\
\emph{Input:} A set of $n$ red strings in~$\{0,1\}^d$ and a set of $n$ blue strings in~$\{0,1\}^d$.\\
\emph{Output:} A pair of a red string $S_r$ and a blue string $S_b$, such that the Hamming (edit) distance between $S_r$ and $S_b$ is minimized.
\end{problem}

We will be particularly interested in two lower bounds related to this problem conditional on the Strong Exponential Time Hypothesis (SETH). 

\begin{conjecture}[SETH]
For any $\delta > 0$, there exists $m = m (\delta)$ such that SAT on $m$-CNF formulas with $n$ variables cannot
be solved in time $\Oh(2^{(1-\delta)n})$.
\end{conjecture} 

Alman and Williams~\cite{AW:2015} showed that if there is $\delta > 0$ such that for all constant $c > 0$, the bichromatic closest pair with mismatches problem, with $d = c \log n$, can be solved in $\Oh(n^{2-\delta})$ time, then SETH is false. By a standard reduction, this implies that no algorithm can process a dictionary of $n$ strings of length $d$ in polynomial time, and subsequently answer dictionary look-up queries with $k = \Theta(d)$ mismatches in $\Oh(n^{1-\delta})$ time. Rubinstein~\cite{R:2018} showed that similar lower bounds hold even when approximation is allowed, both for the Hamming and the edit distances.

\begin{theorem}[\cite{R:2018}]\label{th:hardness_BCD_approx}
Assuming SETH, for all constants $\delta, c > 0$ there exists a constant $\eps = \eps(\delta, c)$ 
such that the running time of any $(1+\eps)$-approximate algorithm for bichromatic closest pair with mismatches or differences is $\Omega(n^{2-\delta})$. For mismatches the length of the strings is $d = c\log n$, for differences $d = c \log n \log \log n$.
\end{theorem}

As a corollary, Rubinstein showed a lower bound for approximate dictionary look-ups. In the $(1+\eps)$-approximate dictionary look-up queries one is asked to output one string at distance at most $(1+\eps) k$ from the query string.

\begin{corollary}[\cite{R:2018}]\label{cor:hardness_DL_approx}
Assuming SETH, for all constants $\delta, c > 0$ there exists a constant $\eps = \eps(\delta, c)$ such that no algorithm can preprocess a dictionary of $n$ strings of length $d = c\log n$ (resp., $d = c \log n \log \log n$), and subsequently answer $(1+\eps)$-approximate dictionary look-up queries with $k = \Theta(\log n)$ mismatches (resp., differences) in $\Oh(n^{1-\delta})$ time.
\end{corollary}

\subsection{Our results and techniques.}\label{sec:ourresults}
We show a number of lower bounds for bichromatic closest pair, dictionary look-up, and text indexing. The first part of our lower bounds is lower bounds for the RAM model conditional on SETH, and the second part of our lower bounds assume the pointer machine model. The pointer machine model focuses on the data structures that solely use pointers to access memory locations, i.e. no address calculations are allowed. Similarly to many other lower bound proofs, we use the variant of this model defined in~\cite{C:90}. Consider a reporting problem, such as the text indexing problem for example. Let $U$ be the universe of possible answers (in the case of text indexing, substrings of the text), and let the answer to a query $Q$ be a subset $\mathcal{S}$ of $U$. We assume that the data structure is a rooted tree of constant degree, where each node represents a memory cell that can store one element of $U$. Edges of the tree correspond to pointers between the memory cells. The information beyond the elements of $U$ and the pointers is not accounted for, it can be stored and accessed for free by the data structure. Given a query $Q$, the algorithm must find all the nodes containing the elements of $\mathcal{S}$. It always starts at the root of the tree and at each step can move from a node to its child. The number of explored nodes is a lower bound on the query time, and the size of the tree is a lower bound on the space. We note that all the data structures that we mentioned in the beginning of this section are pointer machine data structures, in other words, our bounds show that in order to develop more efficient solutions one would need to come up with a radically new approach.

\paragraph*{Hard instances for the edit distance.} We start by introducing a deterministic algorithm that we refer to as \emph{stoppers transform} (Section~\ref{sec:stoppers}). The stoppers transform receives a set of binary strings of a fixed length $d$ and outputs a set of strings of length $d'  = \Oh(d \log d)$ such that the edit distance between the transformed strings is equal to the Hamming distance between the original strings. Thanks to the stoppers transform, it will suffice to analyse bichromatic closest pair, dictionary look-up, and text indexing for the Hamming distance, and lower bounds for the edit distance will follow almost automatically. This is a big advantage, as the edit distance is typically much harder to analyse than the Hamming distance. We note that Rubinstein~\cite{R:2018} used a similar approach to derive a lower bound for bichromatic closest pair with differences from a lower bound for bichromatic closest pair with mismatches, however his algorithm was randomised and preserved the distances only approximately. A disadvantage of our approach is that we use an alphabet of size $\Theta(\log d)$ vs. binary in Rubinstein's approach. For the lower bounds conditional on SETH this is negligible as we work with strings of size $\Theta(\log n)$ and therefore need an alphabet of size $\Theta(\log \log n)$, for the pointer machine lower bounds it is more critical as we need an alphabet of size $\Theta(\sqrt{\log n})$.

\paragraph*{Conditional lower bounds.} As a warm-up, we apply the stoppers transform to derive conditional lower bounds for bichromatic closest pair with differences and dictionary look-up with differences from a conditional lower bound for bichromatic closest pair with mismatches~\cite{AW:2015}. Namely, we show that the bichromatic closest pair with differences cannot be solved in strongly subquadratic time (Corollary~\ref{cor:BCD}) and that any data structure for dictionary look-up with differences that can be constructed in polynomial time cannot have strongly sublinear query time (Corollary~\ref{cor:DLD}).  We note that similar lower bounds can be obtained from the lower bound of Rubinstein for the $(1+\eps)$-approximate bichromatic closest pair, however our proof is simpler as it uses only a series of combinatorial reductions instead of the complex distributed PCP framework. These bounds are tight within polylogarithmic factors since bichromatic closest pair can be solved in $\Oh(n^2 d^2)$ time and dictionary look-up with differences can be solved in $\Oh(n d^2)$ time. On the other hand, these lower bounds only hold for $k = \Theta(\log n)$.

\vspace{-1mm}

\paragraph*{Pointer-machine lower bounds.} For many applications $k$ is relatively small, and one might hope to achieve much better bounds for this regime. We show that this is probably not the case by demonstrating a number of pointer-machine lower bounds that give a more precise dependency between the complexity and the number of mismatches (differences). We start by showing a lower bound for the problem of dictionary look-up with $k$ mismatches (Theorem~\ref{th:pointer_lower_HD_dict}). We use the framework introduced by Afshani~\cite{Afshani13}. Initially introduced to show lower bounds for simplex range reporting and related problems, this framework was later used by Afshani and Nielsen~\cite{AN2016} to show a lower bound for dictionary look-up with $k$ don't cares, which, as we mentioned before, is similar but simpler than that with $k$ mismatches. We extend the framework to the case of mismatches as well and show that any data structure for dictionary look-up with $k$ mismatches that has query time $\Oh(d + (\frac{\log n}{k})^k + occ)$ requires $\Omega(c^k n)$ space, for some constant $c > 1$, for all even $\frac{8}{\sqrt 3} \sqrt{\log n} \le k = o(\log n)$. This is the most technical part of our paper and our main contribution. Applying the stoppers transform, we immediately obtain similar lower bounds for the edit distance (Corollary~\ref{cor:pointer-machine-dictionary-ED}). 

\vspace{-1mm}

\paragraph*{Lower bounds for text indexing.} Finally, in Section~\ref{sec:textindexing} we show a reduction from dictionary look-up to text indexing (Theorems~\ref{th:dictionary-textindexing-HAM},~\ref{th:dictionary-textindexing-ED}) which gives us lower bounds for text indexing. The main idea of the reduction is quite simple: We define the text as the concatenation of the dictionary strings interleaved with a special gadget string. The gadget string must guarantee that if we align the pattern with a substring that is not in the dictionary, the Hamming (edit) distance will be much larger than $k$. 

Via this reduction, Corollaries~\ref{cor:SETH-text-indexing-HAM},~\ref{cor:SETH-approx-text-indexing-HAM},~\ref{cor:SETH-text-indexing-ED},~\ref{cor:SETH-approx-text-indexing-ED} imply that for some value of $k$, there is no data structure for text indexing with $k$ mismatches (or $k$ differences) with sublinear query time unless SETH is false. This lower bound holds even if approximation is allowed. To be more precise, we show that for any constant $\delta > 0$, there exists $\eps = \eps(\delta)$ such that any data structure for $(1+\eps)$-approximate text indexing with mismatches (or differences) that can be constructed in polynomial time requires $\Omega(n^{1-\delta})$ query time. In these lower bounds we assume that the text index outputs just one substring of the text. See a diagram of the reductions in Fig.~\ref{fig:SETH}. 

Corollary~\ref{cor:pointer-machine-text-indexing-HAM} assumes the pointer machine model and gives a more precise dependency of the query time and space on $k$. It shows that any data structure for text indexing with $k$ mismatches that has query time $\Oh(d + (\frac{\log n}{2k})^k + occ)$ must use $\Omega(c^k n)$ space, for some constant $c > 1$. The bound holds for all even $\frac{8}{\sqrt 3} \sqrt{\log n} \le k = o(\log n)$. This gives a partial answer to the open question of Navarro~\cite{ConjectureNavarro}. See a diagram of the reductions in Fig.~\ref{fig:pointer-machine}.


\section{Stoppers transform}\label{sec:stoppers}
In this section we introduce a transform on strings that we refer to as \emph{stoppers transform}. The main property of this transform is that the edit distance between the transformed strings is equal to the Hamming distance between the original strings. We will use this property to derive lower bounds for the edit distance from lower bounds for the Hamming distance. 

We will apply the transform to binary strings of a fixed length $d$. We assume that $d$ is  power of two, $d = \blockdist^{q}$. If this is not the case, we append the strings with an appropriate number of zeros, which does not change the Hamming distance between them and increases the length by a factor of at most two. 

We fix $q$ symbols $c_1,\ldots,c_{q}$ that do not belong to the binary alphabet $\{0,1\}$. For each ${i = 1,2,\ldots,q}$, we define a string $S_i = \underbrace{c_i c_i \ldots c_i}_{6 \cdot \blockdist^i}$. We call $S_i$ a \emph{stopper}. 

Let $X$ be a string $\in \{0,1\}^d$, where $d = 2^q$. First, we insert a stopper $S_q$ after the first $\blockdist^{q-1}$ symbols of $X$. Second, we apply the transform recursively to the strings consisting of the first $\blockdist^{q-1}$ and the last $\blockdist^{q-1}$ symbols of $X$. We summarize the transform in Algorithm~\ref{alg:stopperstransform}.

\begin{algorithm}
\caption{\textsc{StoppersTransform$(X)$}}
\begin{algorithmic}[1]
\If {$q = 0$}
\State \textbf{return} $X$
\Else 
\State $X_1 = $ \textsc{StoppersTransform}$(X[1,\blockdist^{q-1}])$
\State $X_2 = $ \textsc{StoppersTransform}$(X[\blockdist^{q-1}+1,\blockdist^{q}])$
\State $X = X_1 S_q X_2$
\State \textbf{return} $X$
\EndIf
\end{algorithmic}
\label{alg:stopperstransform}
\end{algorithm} 

\begin{fact}\label{fact:length}
Let $\tau(X)$ be the result of applying the stoppers transform to a $d$-length string~$X$. The length of $\tau(X)$ is $\Oh(d \log d)$.
\end{fact}

\begin{theorem}\label{thstopperstransform}
Let $\tau(X), \tau(Y)$ be two strings obtained from binary strings $X, Y$ by applying the stoppers transform. The edit distance between $\tau(X)$ and $\tau(Y)$ is equal to the Hamming distance between $X$ and $Y$.
\end{theorem}

Before starting the proof of Theorem~\ref{thstopperstransform}, let us first remind the standard notion of an alignment.

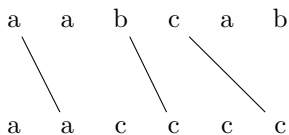
\begin{figure*}[h!]
\begin{center}
\begin{tikzpicture}[scale = 0.7]
\foreach \x [count=\xi] in {a,a,b,c,a,b} {
	\node[above] at (\xi, 2) (s-\xi) {\x};
}

\foreach \x [count=\xi] in {a,a,c,c,c,c} {
	\node[above] at (\xi, 0) (t-\xi) {\x};
}

\draw (s-1) edge (t-2);
\draw (s-3)--(t-4);
\draw (s-4)--(t-6);
\end{tikzpicture}
\end{center}
\caption{An alignment between $T = aabcab$ and $S = aacccc$ of cost $7$.}
\label{fig:alignment}
\end{figure*}

\begin{definition}[Alignment]
Given two strings $T = t_1 t_2 \ldots t_{d_1}$ and $S = s_1 s_2 \ldots s_{d_2}$, an alignment is a partial function $f: [1,d_1] \rightarrow [1,d_2]$. If for $i < j$ both $f(i)$ and $f(j)$ are defined, $f(i) < f(j)$. An alignment can be represented as a bipartite graph, where the nodes are $t_1, t_2, \ldots, t_{d_1}$ and $s_1, s_2, \ldots, s_{d_2}$ respectively, and two nodes $s_i, t_j$ are connected with an edge iff $f(i) = j$. Each node has at most one outgoing edge and the edges do not cross. (See Fig.~\ref{fig:alignment}.)

The cost of an alignment is defined to be equal to the number of nodes that do not have outgoing edges plus the number of edges that connect $s_i, t_j$ such that $s_i \neq t_j$. We say that the alignment is \emph{optimal} if it has the minimal cost. Note that, by definition, the minimal cost is equal to the edit distance between $T$ and $S$. 
\end{definition}

\begin{proof}
Let $\ED(\tau(X), \tau(Y))$ be the edit distance between $\tau(X)$ and $\tau(Y)$, and $\Ham(X,Y)$ be the Hamming distance between $X$ and $Y$. Let us first provide an upper bound for $\ED(\tau(X), \tau(Y))$. We align the $i$-th symbol of $\tau(X)$ with the $i$-th symbol of $\tau(Y)$. The cost of this alignment is equal to the Hamming distance between $X$ and $Y$: The stoppers contribute $0$ to the Hamming distance (because we always align $S_1$ with $S_1$, $S_2$ with $S_2$, etc.), and the symbols of $X,Y$ contribute $1$ whenever there is a mismatch and $0$ otherwise. Therefore, the cost of the alignment is $\Ham(X,Y)$, and $\ED(\tau(X), \tau(Y))$ can only be smaller.

We now show that $\ED(\tau(X),\tau(Y)) \ge \Ham(X,Y)$, which gives the claim of the theorem. Recall that $d = \blockdist^q$. We prove the claim by induction on $q$. If $q = 0$, the claim follows immediately. We now assume that the claim holds for $q_0$, and show that it holds for $q = q_0 + 1$ as well. From the inductive hypothesis it follows that the claim holds for the substrings of $X[1,2^{q-1}]$ and $Y[1,2^{q-1}]$, and $X[2^{q-1}+1,2^q]$ and $Y[2^{q-1}+1,2^q]$. 

Consider an optimal alignment of $\tau(X)$ and~$\tau(Y)$, and let $M$ be the maximal-length contiguous substring of $\tau(Y)$ such that each of its symbols is either aligned with a symbol of the copy of~$S_q$ in~$\tau(X)$, or is not aligned with any of the symbols of~$\tau(X)$. We claim that $M$ contains at least a $2/3$-fraction of the symbols of the copy of $S_q$ in $\tau(Y)$. Indeed, suppose that it is not the case, then the copy of $S_q$ in $\tau(Y)$ contains at least $|S_q| / 3$ symbols that are either deleted, or aligned with symbols in $\tau(X[1,2^{q-1}])$ or $\tau(X[2^{q-1}+1,2^q])$. Since all the symbols in $\tau(X[1,2^{q-1}])$ or $\tau(X[2^{q-1}+1,2^q])$ are different from~$c_q$, the cost of the alignment of $\tau(X)$ and $\tau(Y)$ is at least $|S_q| / 3 = 2 \cdot \blockdist^q > d \geq \Ham (X,Y)$, which contradicts the optimality of the alignment. In particular, it follows that $M$ and the copy of $S_q$ in $\tau(Y)$ have a non-empty overlap.

Consider the factorisations $\tau(X) = X_1 S_q X_2$ and $\tau(Y) = Y_1 M Y_2$, where $S_q$ is aligned with $M$, $X_1$~is aligned with $Y_1$, and $X_2$ is aligned with $Y_2$. From above, it follows that there are four possible cases depending on the location of $M$. Let $\Delta_1, \Delta_2$ be the distances from the endpoints of $M$ to the endpoints of $S_q$ (see Fig.~\ref{fig:stoppers}). 

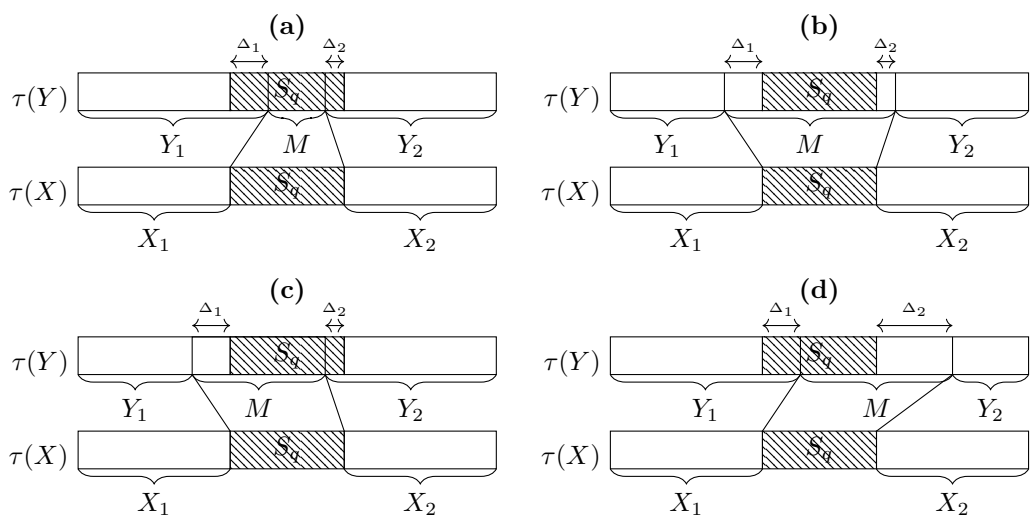
\begin{figure*}
\usetikzlibrary{arrows.meta, decorations.pathreplacing, shadows}
\begin{center}
\begin{tikzpicture}[scale=0.5]
\begin{scope}[yshift=-7cm]
\draw[pattern=north west lines] (0,2.5) rectangle (3,3.5) node[pos=0.5] {$S_q$};
\draw (-4,2.5) rectangle (0,3.5);
\draw (3,2.5) rectangle (7,3.5);
\draw[pattern=north west lines] (0,0) rectangle (3,1)  node[pos=0.5] {$S_q$};
\draw (-4,0) rectangle (0,1);
\draw (3,0) rectangle (7,1);
\draw (0,1)--(-1,2.5);
\draw (3,1)--(2.5,2.5);
\draw (-1,2.5) rectangle (2.5,3.5) ;
\draw[<->] (2.5,3.8)--(3,3.8) node[midway,above] {\tiny{$\Delta_2$}};
\draw[<->] (-1,3.8)--(0,3.8) node[midway,above] {\tiny{$\Delta_1$}};
\node at (1.5,4.7) {\textbf{(c)}};

\node[left] at (-4,0.3) {$\tau(X)$};
\node[left] at (-4,2.8) {$\tau(Y)$};

\draw [decorate, decoration={brace, mirror, amplitude=2mm}] (-4,0) -- (0,0) node [midway, below=2mm] {$X_1$};
\draw [decorate, decoration={brace, mirror, amplitude=2mm}] (3,0) -- (7,0) node [midway, below=2mm] {$X_2$};
\draw [decorate, decoration={brace, mirror, amplitude=2mm}] (-4,2.5) -- (-1,2.5) node [midway, below=2mm] {$Y_1$};
\draw [decorate, decoration={brace, mirror, amplitude=2mm}] (-1,2.5) -- (2.5,2.5) node [midway, below=2mm] {$M$};
\draw [decorate, decoration={brace, mirror, amplitude=2mm}] (2.5,2.5) -- (7,2.5) node [midway, below=2mm] {$Y_2$};
\end{scope}

\begin{scope}[yshift=-7cm,xshift=14cm]
\draw[pattern=north west lines] (0,2.5) rectangle (3,3.5) node[pos=0.5] {$S_q$};
\draw (-4,2.5) rectangle (0,3.5);
\draw (3,2.5) rectangle (7,3.5);
\draw[pattern=north west lines] (0,0) rectangle (3,1)  node[pos=0.5] {$S_q$};
\draw (-4,0) rectangle (0,1);
\draw (3,0) rectangle (7,1);

\node[left] at (-4,0.3) {$\tau(X)$};
\node[left] at (-4,2.8) {$\tau(Y)$};

\draw [decorate, decoration={brace, mirror, amplitude=2mm}] (-4,0) -- (0,0) node [midway, below=2mm] {$X_1$};
\draw [decorate, decoration={brace, mirror, amplitude=2mm}] (3,0) -- (7,0) node [midway, below=2mm] {$X_2$};
\draw [decorate, decoration={brace, mirror, amplitude=2mm}] (-4,2.5) -- (1,2.5) node [midway, below=2mm] {$Y_1$};
\draw [decorate, decoration={brace, mirror, amplitude=2mm}] (1,2.5) -- (5,2.5) node [midway, below=2mm] {$M$};
\draw [decorate, decoration={brace, mirror, amplitude=2mm}] (5,2.5) -- (7,2.5) node [midway, below=2mm] {$Y_2$};

\draw (0,1)--(1,2.5)--(1,3.5);
\draw (3,1)--(5,2.5)--(5,3.5);
\draw[<->] (0,3.8)--(1,3.8) node[midway,above] {\tiny{$\Delta_1$}};
\draw[<->] (3,3.8)--(5,3.8) node[midway,above] {\tiny{$\Delta_2$}};
\node at (1.5,4.7) {\textbf{(d)}};
\end{scope}

\begin{scope}[]
\draw[pattern=north west lines] (0,2.5) rectangle (3,3.5) node[pos=0.5] {$S_q$};
\draw (-4,2.5) rectangle (0,3.5);
\draw (3,2.5) rectangle (7,3.5);
\draw[pattern=north west lines] (0,0) rectangle (3,1)  node[pos=0.5] {$S_q$};
\draw (-4,0) rectangle (0,1);
\draw (3,0) rectangle (7,1);

\node[left] at (-4,0.3) {$\tau(X)$};
\node[left] at (-4,2.8) {$\tau(Y)$};

\draw [decorate, decoration={brace, mirror, amplitude=2mm}] (-4,0) -- (0,0) node [midway, below=2mm] {$X_1$};
\draw [decorate, decoration={brace, mirror, amplitude=2mm}] (3,0) -- (7,0) node [midway, below=2mm] {$X_2$};
\draw [decorate, decoration={brace, mirror, amplitude=2mm}] (-4,2.5) -- (1,2.5) node [midway, below=2mm] {$Y_1$};
\draw [decorate, decoration={brace, mirror, amplitude=2mm}] (1,2.5) -- (2.5,2.5) node [midway, below=2mm] {$M$};
\draw [decorate, decoration={brace, mirror, amplitude=2mm}] (2.5,2.5) -- (7,2.5) node [midway, below=2mm] {$Y_2$};

\draw (0,1)--(1,2.5)--(1,3.5);
\draw (3,1)--(2.5,2.5)--(2.5,3.5);
\draw[<->] (0,3.8)--(1,3.8) node[midway,above] {\tiny{$\Delta_1$}};
\draw[<->] (2.5,3.8)--(3,3.8) node[midway,above] {\tiny{$\Delta_2$}};
\node at (1.5,4.7) {\textbf{(a)}};
\end{scope}

\begin{scope}[xshift=14cm]
\draw[pattern=north west lines] (0,2.5) rectangle (3,3.5) node[pos=0.5] {$S_q$};
\draw (-4,2.5) rectangle (0,3.5);
\draw (3,2.5) rectangle (7,3.5);
\draw[pattern=north west lines] (0,0) rectangle (3,1)  node[pos=0.5] {$S_q$};
\draw (-4,0) rectangle (0,1);
\draw (3,0) rectangle (7,1);

\node[left] at (-4,0.3) {$\tau(X)$};
\node[left] at (-4,2.8) {$\tau(Y)$};

\draw [decorate, decoration={brace, mirror, amplitude=2mm}] (-4,0) -- (0,0) node [midway, below=2mm] {$X_1$};
\draw [decorate, decoration={brace, mirror, amplitude=2mm}] (3,0) -- (7,0) node [midway, below=2mm] {$X_2$};
\draw [decorate, decoration={brace, mirror, amplitude=2mm}] (-4,2.5) -- (-1,2.5) node [midway, below=2mm] {$Y_1$};
\draw [decorate, decoration={brace, mirror, amplitude=2mm}] (-1,2.5) -- (3.5,2.5) node [midway, below=2mm] {$M$};
\draw [decorate, decoration={brace, mirror, amplitude=2mm}] (3.5,2.5) -- (7,2.5) node [midway, below=2mm] {$Y_2$};

\draw (0,1)--(-1,2.5)--(-1,3.5);
\draw (3,1)--(3.5,2.5)--(3.5,3.5);
\draw[<->] (3,3.8)--(3.5,3.8) node[midway,above] {\tiny{$\Delta_2$}};
\draw[<->] (-1,3.8)--(0,3.8) node[midway,above] {\tiny{$\Delta_1$}};

\node at (1.5,4.7) {\textbf{(b)}};
\end{scope}
\end{tikzpicture}
\end{center}
\caption{\label{fig:stoppers}Four possible locations of $M$ in $\tau(Y)$.}
\end{figure*}

\textbf{Case (a).} The edit distance between $S_q$ and $M$ is $\Delta_1 + \Delta_2$. On the other hand, the edit distance between $X_1 = \tau(X[1,2^{q-1}])$ and $Y_1$ is at least the edit distance between $X_1$ and $\tau(Y[1, 2^{q-1}])$ minus $\Delta_1$ because of the triangle inequality:
\begin{align*}
\ED(X_1,Y_1) \ge \ED(X_1,\tau(Y[1,2^{q-1}])) - \ED(Y_1,\tau(Y[1,2^{q-1}]))  = \ED(X_1,\tau(Y[1,2^{q-1}])) - \Delta_1
\end{align*}
Here the last equality follows from the fact that $Y_1$ is equal to $\tau(Y[1,2^{q-1}]) S_q[1,\Delta_1]$.
Symmetrically, the edit distance between $X_2 = \tau(X[2^{q-1}+1,2^q])$ and $Y_2$ is at least the edit distance between $X_2$ and $\tau(Y[2^{q-1}+1, 2^q])$ minus $\Delta_2$. Therefore, by the inductive hypothesis, the edit distance between $\tau(X)$ and $\tau(Y)$ is at least $\Ham (X,Y)$.

\textbf{Case (b).} The edit distance between $S_q$ and $M$ is at least $\Delta_1 + \Delta_2$. On the other hand, the edit distance between $X_1 = \tau(X[1,2^{q-1}])$ and $Y_1$ is at least the edit distance between $X_1$ and $\tau(Y[1, 2^{q-1}])$ minus $\Delta_1$:
\begin{align*}
\ED(X_1,Y_1) \ge \ED(X_1,\tau(Y[1, 2^{q-1}])) - \ED(Y_1,\tau(Y[1, 2^{q-1}])) = \ED(X_1,\tau(Y[1, 2^{q-1}])) - \Delta_1
\end{align*}
Here the last equality follows from the fact that $Y_1$ is equal to $\tau(Y[1, 2^{q-1} - \Delta_1])$. Symmetrically, the edit distance between $X_2$ and $Y_2$ is at least the edit distance between $X_2 = \tau(X[2^{q-1}+1,2^q])$ and $\tau(Y[2^{q-1}+1, 2^q])$ minus $\Delta_2$. Therefore, by the inductive hypothesis, the edit distance between $\tau(X)$ and $\tau(Y)$ is at least $\Ham (X,Y)$.

\textbf{Cases (c) and (d).} These two cases are symmetrical, so it suffices to consider case (c). The edit distance between $S_q$ and $M$ is at least $\Delta_1$, because of the different alphabets. As before, the edit distance between $X_1$ and $Y_1$ is at least the edit distance between $X_1$ and $\tau(Y[1, 2^{q-1}])$ minus $\Delta_1$. On the other hand, the edit distance between $X_2$ and $Y_2$ must be at least the edit distance between $X_2$ and $\tau(Y[2^{q-1}+1, 2^q])$. Indeed, recall that an optimal alignment of $X_2$ and $Y_2$ can be considered as a bipartite graph on the symbols of $X_2$ and $Y_2$. We can then take an alignment between $X_2$ and $\tau(Y[2^{q-1}+1, 2^q])$ to be a subgraph of the optimal alignment of  $X_2$ and $Y_2$ induced by the symbols of $X_2$ and of $\tau(Y[2^{q-1}+1, 2^q])$. The cost of the alignment will not change, as the alphabets of $S_q$ and $X_2$ are different and aligning any two symbols of them costs one.  
The claim follows by the inductive hypothesis.
\end{proof}

\subsection{First application of the stoppers transform.}
As a warm-up, we show the first application of the stoppers transform: conditional lower bounds for the problems of bichromatic closest pair and of dictionary look-up with differences. As we have mentioned, similar lower bounds follow from Theorem~\ref{th:hardness_BCD_approx} and Corollary~\ref{cor:hardness_DL_approx}, however, the stoppers transform allows to obtain these bounds in a more direct way avoiding the complex distributed PCP framework. 

Alman and Williams~\cite{AW:2015} showed the following conditional lower bound for the bichromatic closest pair with mismatches problem:

\begin{theorem}[\cite{AW:2015}]\label{th:BCD-mismatches}
Assuming SETH, for every $\delta > 0$ there exists a constant ${c = c(\delta)}$ such that the problem of bichromatic closest pair with mismatches for two sets of strings in $\{0,1\}^{c \log n}$ of size $n$ each cannot be solved in time $\Oh( n^{2-\delta})$. 
\end{theorem} 

\begin{corollary}\label{cor:BCD}
 Assuming SETH, for every $\delta > 0$ there exists a constant $c' = c'(\delta)$ such that the problem of bichromatic closest pair with differences for two sets of $c' \log n \log \log n$-length strings of size $n$ each cannot be solved in time $\Oh( n^{2-\delta})$.
\end{corollary}
\begin{proof}
By Theorem~\ref{th:BCD-mismatches}, for every $\delta > 0$ there is a hard instance of bichromatic closest pair with mismatches that cannot be solved in $\Oh(n^{2-\delta})$ time. The length of the strings in this instance is equal to $c \log n$ for some constant $c$. We apply the stoppers transform to each of the strings in this hard instance. As a result, the length of the strings becomes $c' \log n \log \log n$ for some constant~$c'$, and the edit distance between any two strings is equal to the Hamming distance between their preimages. Therefore, if we can solve the bichromatic closest pair with differences problem in $\Oh(n^{2-\delta})$ time, we can solve the hard instance of the  bichromatic closest pair with mismatches problem in $\Oh(n^{2-\delta})$ time, a contradiction. 
\end{proof}

We now show a lower bound for dictionary look-up with differences via a standard reduction.

\begin{corollary}\label{cor:DLD}
Assuming SETH, for every $\delta > 0$ there exists $c' = c'(\delta)$ and $k \le c' \log n \log \log n$ such that any data structure with polynomial construction time that solves the dictionary look-up with $k$ differences problem for a set of $n$ strings of length $c' \log n \log \log n$ cannot have query time $\Oh( n^{1-\delta})$. 
\end{corollary}
\begin{proof}
Let $\delta' < \delta$ be a constant. By Corollary~\ref{cor:BCD}, there is a hard instance of bichromatic closest pair with differences that cannot be solved in time $\Oh(n^{2-\delta'})$.
Suppose first that the construction time of the data structure for dictionary look-up is $n^\gamma = o(n^{2-\delta'})$. We can then use it to solve the bichromatic closest pair with differences problem in the following way: For each $k =  1, 2, \dots, c'\log n \log \log n$ we build the data structure for the red strings, and then run dictionary look-up with $k$ differences for each $k$ and for each of the blue strings. In total, there are $\Oh(n \log n \log \log n)$ queries and  therefore, at least one of these queries cannot be answered in time $\Oh(n^{1-\delta})$.

Suppose now that the construction time of the data structure is $n^\gamma = \Omega(n^{2-\delta'})$. We use a standard approach to overcome this technicality. We divide the set of the red strings into $n^{1-1/2\gamma}$ subsets containing $n^{1/2\gamma}$ strings each. For each of the subsets and for each ${k = 1, 2, \dots, c'\log n \log \log n}$, we build the data structure for dictionary look-up. The total construction time is $n^{3/2-1/2\gamma} c'\log n \log \log n = o(n^{3/2})$. We then run a dictionary look-up with $k$ differences for each $k = 1, 2, \dots, c'\log n \log \log n$ and for each of the blue strings. In total, we run $\Oh(n^{2-1/2\gamma} \log n \log \log n)$ queries. Therefore, at least one of these queries cannot be answered in time $\Oh(n^{1/2\gamma-\delta})$. Recall that each of the data structures contains $n^{1/2\gamma}$ strings. The claim follows.
\end{proof}

The size of the alphabet in the hard instances of bichromatic closest pair and dictionary look-up with differences that we construct in Corollaries~\ref{cor:BCD} and~\ref{cor:DLD} is $\Theta(\log \log n)$.
\section{Lower bounds for dictionary look-up}
In this section, we show a  lower bound for dictionary look-up in the pointer-machine model, that will later be used to show a lower bound for text indexing. 

\subsection{Theorem and proof outline}

\begin{theorem}\label{th:pointer_lower_HD_dict}
Assuming the pointer-machine model, for any even $\frac{8}{\sqrt{3}} \sqrt{\log n} \le k = o(\log n)$ there exists~$d > 0$ such that any data structure for dictionary look-up with $k$ mismatches on a set of strings in $\{0,1\}^d$ of total length $n$ with query time $\Oh(d + \left( \frac{\log n}{k} \right)^k + occ)$ must use $\Omega \left(n\cdot c^{k/16}\right)$ space, for some constant $c > 2$.
\end{theorem}

Note that the lower bound of Theorem~\ref{th:pointer_lower_HD_dict} is more precise compared to what we gave in Section~\ref{sec:introduction} (there we said that the lower bound for space is $\Omega(n c_1^{k})$ for some constant $c_1 > 1$), we will use it later to derive the lower bounds for text indexing.

Similar to~\cite{AN2016}, we use the framework of Afshani~\cite{Afshani13} that gives a pointer machine lower bound for a problem called \emph{geometric stabbing}. 
The first step is to design our lower bound instances, that is, we define a dictionary and a set of queries, and then translate them into geometric objects. 
The main part of the proof consists in proving that these objects satisfy the hypothesis of the theorem of~\cite{Afshani13}. 
Finally we apply the theorem and translate the conclusion to our original context.
Proving that the condition of the theorem is met heavily relies on precisely understanding how the Hamming distance behaves on the special instances we consider. 

\subsection{Geometric stabbing lower bound.}

In this subsection, we describe the geometric lower bound of~\cite{Afshani13}. 
Consider the $d$-dimensional Hamming cube. In the geometric stabbing problem, we are given a subset $\Q$ of points, and a set $\mathcal{R}$ of geometric regions in this cube. We must preprocess $\mathcal{R}$ into a data structure, to support the following queries: given a point $q \in \Q$, output all the regions that contain $q$. 
Theorem~1 of Afshani~\cite{Afshani13} gives the following pointer machine lower bound for this problem:

\begin{theorem}[\cite{Afshani13}]\label{th:pointer_lower_bound_gen}
Consider a set $\mathcal{R}$ of $m$ geometric regions that satisfies the following two conditions for some parameters $\beta$, $t$, and $v$: 
\begin{itemize}
\setlength\itemsep{1pt} 
	\item every point in $\Q$ is contained in at least $t$
regions;
	\item the volume of the intersection of every $\beta$ regions is at most $v$.
\end{itemize}
Then for any pointer-machine data structure, if answering geometric stabbing queries can be done in time $g(m)+\Oh(output)$, where $g$ is an increasing function and $g(m) \le t$, then the space used is $S(m) = \Omega (t v^{-1} 2^{-\Oh(\beta)})$.
\end{theorem}

A modification, that already appeared in~\cite{AN2016}, has to be done here.
First, in the theorem the volume is implicitly measured by the classic geometric measure (that is, the Lebesgue measure), and the regions are simplices, however, any reasonable measure and any reasonable definition of a region also works.
Furthermore, the original theorem in~\cite{Afshani13} states that every point is contained in \emph{exactly} $t$ regions, but again, the proof also works for \emph{at least} $t$ regions as in the version above. Finally, we stated the theorem in full generality, but we will only use it for $\beta=2$, thus the space lower bound is $\Omega (t/v)$.

Throughout this proof we will use extensively, and without explicitly reminding it, the following well-known fact about binomial coefficients.

\begin{fact}
For all $n>k>0$ $\left(\frac{n}{k}\right)^k \leq \binom{n}{k} \leq \left(\frac{n\cdot e}{k}\right)^k$.
\end{fact}

\subsection{Definitions of \texorpdfstring{$\Q$}{Q} and of the measure.}
We now define the queries used in the lower bound. These are binary strings of length $d$, where $d$ is a parameter that we will define later. We can use the natural bijection to map these strings into points in the $d$-dimensional Hamming cube. 
We then directly get the query points for the geometric stabbing problem.
(In the following, we may use strings and points interchangeably when it is not misleading.)
A string $P \in \Q$ if, when we divide it into $k$ blocks of lengths $d/k$ (we assume that $d$ is a multiple of $k$), $k/2$ blocks contain exactly one set bit, and $k/2$ blocks contain exactly two set bits. This is where we use the fact that $k$ is even. Note that one can generate these strings by first choosing which blocks contain two set bits, then choosing one set bit in every block, and finally choosing one additional set bit in $k/2$ blocks. 
We now obtain the following fact:

\begin{fact}\label{fact:Q}
The size of $\Q$ is ${\dbinom{k}{k/2}} (d/k)^k (d/k-1)^{k/2}$.
\end{fact}

The measure we will use when applying Theorem~\ref{th:pointer_lower_bound_gen} is as follows: 
Each string in $\Q$ has measure~$1/|\Q|$, and the rest of the cube has measure $0$.

\subsection{Definition of \texorpdfstring{$\D$}{D} and of the regions.}
Next, we select the set of strings of the dictionary and the associated set of regions. The dictionary is a set of strings $\D$ that we will define soon. The associated regions are the Hamming balls of radius $k$ whose centres are the strings of $\D$. Note that reporting the set of regions a string belongs to is equivalent to reporting the set of dictionary strings that are at distance at most $k$ from this string. The definition of $\D$ uses the following quantities:

\[
\alpha = 
\log_{d/k} \frac{\log n}{k} \ \ \text{ and}\ \ 
p_\alpha = 
\frac{1}{\sum_{i = 0}^{\lfloor{k \left(\frac{1}{4} - \alpha\right)\rfloor}} \dbinom{d}{i}}
\]
(We will show later that $\alpha < 1/4$ and therefore these values are defined correctly.) Below we will show the following lemma:

\begin{lemma}\label{lem:regions}
For any $\frac{8}{\sqrt{3}} \sqrt{\log n} \le k = o(\log n)$, there is a set $\D$, and an associated set of regions $\mathcal{R}$ that are the Hamming balls of radius $k$ centered at points in $\D$, such that:
\begin{enumerate}
\item The Hamming distance between each two strings in $\D$ is at least $\lfloor{ k (\frac{1}{4} - \alpha)\rfloor}$;
\item $\D$ has size $\Theta (p_\alpha \cdot(d/k)^{k})$;
\item Each string in $\Q$ belongs to at least $t = \Omega\left( p_\alpha 2^{k/4} {\dbinom{k}{k/4}} (d/k-2)^{k/4} \right)$ regions of $\mathcal{R}$.
\end{enumerate}
\end{lemma}

\paragraph*{Definition of $\D$.} 
We define $\D$ probabilistically and prove that it satisfies the three properties of Lemma~\ref{lem:regions} with constant probability. By the probabilistic method, this implies that there exists a set $\D$ with these properties. 

We define $\D$ in two steps. First, we construct the set of strings of length $d$ such that if we divide it into blocks of length $d/k$, each block will contain exactly one set bit. Note that this is similar to the shape of the query strings but is not the same. This set has size $(d/k)^k$. Second, we sparsify this set. We first filter the set of strings defined above by selecting each string in this set with probability~$p_\alpha$, and then if there are two strings at distance at most $\lfloor{ k (1/4-\alpha) \rfloor}$ from each other, we delete them both. The remaining strings form the set $\D$, and the set of Hamming balls of radius~$k$ centered at points in $\D$ is called~$\mathcal{R}$. Note that Property 1 of Lemma~\ref{lem:regions} is satisfied by definition.

Let us show that all the steps are correctly defined, i.e. that $p_\alpha$ is well-defined.

\begin{claim}\label{claim:alpha-zero}
If $|\D| \le (d/k)^k$, then $\alpha$ tends towards zero.
\end{claim}  
\begin{proof}
Let us first show that $\log d -\log k =\Omega(\log n/k)$.
As $n$ is the total length of the strings in the dictionary, and it contains $|\D|$ strings of length~$d$, we have $n=d|\D|$. Therefore, $n\leq d (d/k)^k$. By taking the logarithm and isolating $\log d$, we get $\log d \geq (\log n + k\log k)/(k+1)$.
Then $\log d - \log k \geq (\log n - \log k)/(k+1)$ and asymptotically the right hand term is in $\Omega(\log n/k)$ because $k$ is negligible compared to $n$. (Recall that we assume $\frac{8}{\sqrt{3}} \sqrt{\log n} \le k = o(\log n)$ throughout Lemma~\ref{lem:regions}.)

Remember that $\alpha = \log_{d/k} \frac{\log n}{k} = (\log \log n - \log k) / (\log d - \log k)$ and therefore we obtain that $\alpha = \Oh(k (\log \log n - \log k) / \log n)$. Let $h=\log n/k$, then:

\[
\frac{k (\log \log n - \log k)}{\log n}=\frac{\log h }{h}
\]
As $k=o(\log n)$, we know that $h$ tends toward infinity, therefore $\log h/h$ tends to zero. The claim follows.
\end{proof}

\paragraph*{Size of $\D$.} Next, we show that Property 2 of Lemma~\ref{lem:regions} holds with constant probability.

\begin{claim}\label{claim:size-D}
There exist two constants $c_1<c_2$ such that $|\D|$ has size between $c_1p_\alpha(d/k)^k$ and $c_2p_\alpha(d/k)^k$ with probability at least $ 2/3$. 
\end{claim}
\begin{proof}
Let us first estimate the probability of a string $S$ to make it to the set $\D$, provided it has one set bit per block. 
The string $S$ survives the first filtering step with probability $p_\alpha$. Note that $1/p_\alpha$ is the number of strings in any Hamming ball of radius $\lfloor{ k \left(\frac{1}{4} - \alpha\right) \rfloor}$. Therefore, after this step, the expected number of the selected strings in any ball of Hamming radius $\lfloor{ k \left(\frac{1}{4} - \alpha\right) \rfloor}$ is at most~$1$. This implies that the probability that a ball of Hamming radius $\lfloor{ k \left(\frac{1}{4} - \alpha\right) \rfloor}$ with the center at $S$ contains another string can be bounded by $1/2$, by Markov's inequality. Hence, the string $S$ survives the two filtering steps with probability at least $p_{\alpha}/2$, and at most $p_{\alpha}$. It follows that there exist two constants $c_1<c_2$, such that $|\D|$ has size between $c_1 p_\alpha (d/k)^k$ and $c_2 p_\alpha (d/k)^k$ with probability at least $ 2/3$. 
\end{proof}

\paragraph*{Covering the query points by the regions.} Finally, we show Property 3 of Lemma~\ref{lem:regions}.

\begin{claim}\label{claim:LB-regions}
For any string $P \in \Q$ there are $\Omega\left( p_\alpha 2^{k/4} {\dbinom{k}{k/4}} (d/k-2)^{k/4} \right)$ regions containing it.
\end{claim}
\begin{proof}
The number of regions that contain $P$ is the number of strings in $\D$ that are at Hamming distance at most $k$ from $P$.
Consider a string $S$ with one set bit in each block, such that the distance from $P$ to $S$ is at most $k$. Divide all blocks of $S$ into four groups: 

\begin{itemize}
\item $x_{-}$ blocks such that the corresponding block in $P$ also contains only one set bit, but the positions of the set bits in the blocks \emph{do not} match. 
\item $x_+$ blocks such that the corresponding block in $P$ also contains only one set bit, and the positions of the set bits in the blocks \emph{do} match. 
\item $y_{-}$ blocks such that the corresponding block in $P$ contains two set bits, and the position of the set bit in $S$ \emph{does not match any of the positions} of the set bits in $P$.
\item $y_+$ blocks such that the corresponding block in $P$ contains two set bits, and the position of the set bit in $S$ \emph{does match one of the positions} of the set bits in $P$. 
\end{itemize}

Let $\delta$ be the Hamming distance between $P$ and $S$, $\delta=\Ham(P,S)$. For the decomposition above, $\delta = 2x_{-} + 3 y_{-} + y_{+}$. 
On the other hand, $y_{-} + y_+ = k/2$. Therefore, $x_{-} + y_{-} = \delta/2 - k/4$.
The string $S$ can be generated in the following way. 
First choose the $x_{-}$ blocks among the $k/2$ blocks that have only one set bit in $P$. 
This fixes the $x_{+}$ blocks. 
Then choose the $y_{-}=\delta/2 - k/4-x_{-}$ blocks, among the $k/2$ blocks that have two set bits in $P$. 
Then choose the unaligned bit in the $x_{-}$ blocks, copy the set bit in the $x_{+}$ blocks, choose the unaligned set bit in the $y_{-}$ blocks, and choose one of the two set bits in the $y_{+}=3k/4+x_{-}-\delta/2$ blocks.
The number of such strings $S$ is $\sum_{\delta=k/2}^k \sum_{x_{-}=0}^{\delta/4} A_{x_{-}, \delta} B_{x_{-}, \delta}$, where
\begin{align*}
&A_{x_{-}, \delta} = {\dbinom{k/2}{x_{-}}} {\dbinom{k/2}{\delta/2 - k/4 - x_{-}}}\\
&B_{x_{-}, \delta} = (d/k-1)^{x_{-}} (d/k-2)^{\delta/2-k/4-x_{-}} 2^{3k/4+x_{-}-\delta/2}
\end{align*}
By taking only the term $\delta=k$, and lower bounding the terms in $B_{x_{-}, \delta}$, we get the following lower bound:

\[
\sum_{x_{-}=0}^{k/4} 
{\dbinom{k/2}{x_{-}}} {\dbinom{k/2}{k/4 - x_{-}}} 
(d/k-2)^{k/4} 2^{k/4}
\]
We use Vandermonde's identity for the summation of binomials, and get the following lower bound:

\[
{\dbinom{k}{k/4}} (d/k-2)^{k/4} 2^{k/4}
\]
Each such string belongs to the set $\D$ with probability $\Theta(p_\alpha)$. Therefore, the expected number of strings in $\Q$ that are at Hamming distance $k$ from $P$ is  $\Omega\left(p_\alpha {\dbinom{k}{k/4}} 2^{k/4} (d/k-2)^{k/4}\right)$. 
We can lower-bound $p_\alpha$ by $1/\left[k/4 \cdot {\dbinom{d}{\lfloor{k \left(\frac{1}{4} - \alpha\right)\rfloor}}}\right]$. 
Furthermore, ${\dbinom{d}{\lfloor{k \left(\frac{1}{4} - \alpha\right)\rfloor}}} \le (16 d/k)^{k (\frac{1}{4} - \alpha)}$ (here we use the fact that $\alpha$ tends to zero and therefore for $n$ large enough $\lfloor{ k \left(\frac{1}{4} - \alpha\right)\rfloor} \ge 3k/16$). Finally, we have ${\dbinom{k}{k/4}} \ge k$. Summarising, we obtain
 
\[
\Omega\left(p_\alpha {\dbinom{k}{k/4}} 2^{k/4} (d/k-2)^{k/4}\right) = \Omega((d/k)^{\alpha k} / \gamma^k)
\]
for some constant $\gamma$. By the choice of $\alpha$, $(d/k)^{\alpha k} = (\log n / k)^{k}$, so the bound boils down to $\Omega((\log n/ \gamma k)^{k})$.
We now apply the multiplicative Chernoff bound: For $\lambda$ independent random variables $X_1, \ldots , X_\lambda$ in $\{0,1\}$, let $X$ be their sum and the expectation of $X$ be $\mu$, then we have
 
\[
\prob {X \le (1-\delta)\mu} \le e^{-\mu \delta^2/2}
\]
for every $0<\delta < 1$. In our case, $X_i$ is equal to $1$ if the $i$-th point survives and to $0$ otherwise, and $\mu = \Omega((\log n /\gamma k)^k)$ and therefore the number of strings in $\D$ that are at Hamming distance $k$ from $P$ is $\Omega\left(p_\alpha {\dbinom{k}{k/4}} 2^{k/4} (d/k-2)^{k/4}\right)$ with probability 
$\ge 1 - \exp \left[-\Omega\left((\log n / \gamma k)^{k}\right)\right]$. 
By the union bound, this holds for all strings in $\Q$ with probability at least $2/3$. (From Fact~\ref{fact:Q} we have $|\Q| < n^3$.)
\end{proof}

By Claims~\ref{claim:alpha-zero},~\ref{claim:size-D} and~\ref{claim:LB-regions}, the set $\D$ of centres we defined satisfies all the three properties of Lemma~\ref{lem:regions} with constant probability, thus such a set $\D$ exists.

\subsection{Intersection of regions.}
We will now consider the parameter $v$ in Theorem~\ref{th:pointer_lower_bound_gen}. We take the parameter $\beta$ equal to $2$, and consider two strings $S_1, S_2\in \D$. 
From Lemma~\ref{lem:regions} we know that the Hamming distance between them is at least $\lfloor{ k (\frac{1}{4} - \alpha)\rfloor} > k/8$ (remember that from Claim~\ref{claim:alpha-zero}, $\alpha$ tends to zero). The parameter~$v$ of the theorem is the maximum fraction of elements of~$\Q$ that belong to the intersection of spheres of $S_1$ and~$S_2$.

\begin{lemma}\label{lem:fraction-Q}
We have $v |\Q| \le k^2 \times \frac{k}{4} {\dbinom{k/8}{k/16}} {\dbinom{k}{3k/16}} {\dbinom{k}{k/2}} (d/k)^{3k/4}$.
\end{lemma}
\begin{proof}
For two binary strings $A$ and~$B$, define $A \oplus B$ to be the sum element-wise modulo~$2$ of $A$ and $B$. 
Note that the Hamming distance between $A$ and $B$ is the number of set bits in $A \oplus B$.

Now let $S_1, S_2 \in \D$, and $\Ham(S_1, S_2) = 2z > k/8$. 
Consider a string $P \in \Q$ at distance at most~$k$ both from $S_1$ and $S_2$. 
Let $\delta_1=\Ham(S_1,P)$, and $\delta_2=\Ham(S_2,P)$. 
Note that $\delta_1+\delta_2\geq 2z$, and~$|\delta_1-\delta_2|\leq 2z$.
Also, as $S_i$ has one set bit in each block and because $P$ has two set bits in $k/2$ blocks, we know that $\delta_1, \delta_2 \geq k/2$.

Because $S_1$ and $S_2$ have exactly one set bit per block, each block is either identical in $S_1$ and~$S_2$, or has two positions where $S_1$ and $S_2$ do not match.
The distance between $S_1$ and $S_2$ being~$2z$, we know that $S_1 \oplus S_2$ has $z >  k/16$ blocks that are not empty, and each block contains two set bits.
The following claim provides information about $P \oplus S_1$. 

\begin{claim} \label{claim:p+s1}
The following properties hold:
\begin{enumerate}
\item $P \oplus S_1$ contains $\delta_1$ set bits.
\item  Among the $2z$ set bits of $S_1 \oplus S_2$, $z+(\delta_1-\delta_2)/2$ must match the bits of $P \oplus S_1$, and $z-(\delta_1-\delta_2)/2$ must mismatch.
\item The number of set bits in $P \oplus S_1$ that correspond to set bits in $S_1$ but not in $P$ is $\delta_1/2-k/4$.
\end{enumerate}
\end{claim}

\begin{proof}
We prove each of the properties in turn.
\begin{enumerate}
\item This is because the Hamming distance between $P$ and $S_1$ is $\delta_1$.

\item Let us first note that the Hamming distance between $P \oplus S_1$ and $S_1 \oplus S_2$ is $\delta_2$, because it is the number of set bits in $(P \oplus S_1)\oplus (S_1 \oplus S_2)=P\oplus S_2$. 
Now let $r$ be the number of set bits of $S_1 \oplus S_2$ that match set bits in $P\oplus S_1$. 
We express the distance between $P \oplus S_1$ and $S_1 \oplus S_2$ as a function of $r$. 
This distance is the number of set bits in $P \oplus S_1$, plus the number of set bits in $S_1 \oplus S_2$, minus twice the number of set bits that match, that is $\delta_1+2z-2r$. 
Therefore, $r=z+(\delta_1-\delta_2)/2$.

\item Divide all blocks of $P$ into 4 types: 
\begin{itemize}
\item $x_{-}$ blocks that contain one set bit not matching the set bit in the aligned block of $S_1$;
\item $x_+$ blocks that contain one set bit matching the set bit in the aligned block of $S_1$;
\item $y_{-}$ blocks that contain two set bits, neither of which match the set bit in the aligned block of $S_1$;
\item $y_+$ blocks that contain two set bits, one of which matches the set bit in the aligned block of $S_1$.
\end{itemize}
The number of bits that are set in $S_1$, but not in $P$ is then $x_{-}+y_{-}$. 
Then from $\Ham(P,S_1) = 2x_{-} + 3 y_{-} + y_{+}$, and $y_{-} + y_+ = k/2$, we get that this number is equal to $\delta_1/2-k/4$. 
\end{enumerate}
\end{proof}

We now compute the number of such strings $P$, that is the number of strings in $\Q$ that have distance $\delta_1$ and $\delta_2$ to $S_1$ and $S_2$, respectively. 
We will estimate the size of this set from above. 
It is enough to upper bound the size of the set of the strings that satisfy the three properties of Claim~\ref{claim:p+s1}. 
Also instead of $P$ we will count the number of $P\oplus S_1$ strings, as these two numbers are equal. 
To generate such string, we do the following (that we explain and analyse afterwards):

\begin{enumerate}
\item Consider the set bits in $S_1\oplus S_2$. Choose $z+(\delta_1-\delta_2)/2$ of them to be set bits in $P\oplus S_1$, and the rest (that is $z-(\delta_1-\delta_2)/2$ bits) to be zeros in $P\oplus S_1$.
\item Let $w$ be the number of bits that we have just set to one in $P\oplus S_1$, and that are set bits in~$S_1$. Choose $\delta_1/2-k/4-w$ additional bits that are set bits in $S_1$, but not in $S_1\oplus S_2$, and turn them to $1$ in $P \oplus S_1$.
\item Choose the $k/2$ blocks of $P$ that will contain two set bits. 
\item Choose the set bits in each block of $P \oplus S_1$ so that the total number of set bits is $\delta_1$.
\end{enumerate}

After step 1, Property 2 of Claim~\ref{claim:p+s1} is satisfied. After step 2, Property~3 of Claim~\ref{claim:p+s1} is satisfied. After step 4, Property 1 is satisfied. Also, the structure of $P$ follows the definition of~$\Q$. 
Now let us upper bound the number of ways of performing this generation. 
There are $A_w = \dbinom{2z}{z+(\delta_1-\delta_2)/2}$ possibilities for the first step. 
The number of choices for the second step can be upper bounded by $B_w = \dbinom{k}{\delta_1/2-k/4-w}$ (because there is only one set bit in each of the $k$ blocks of $S_1$). 
For step~3, an upper bound is  $C_w = \dbinom{k}{k/2}$. Finally at step 4, the number of bits left to assign is $\frac{\delta_2}{2}+\frac{k}{4}+w-z$. At that point the blocks where these bits are needed are fixed. Thus an upper bound on the number of choices is $D_w = (d/k)^{\frac{\delta_2}{2}+\frac{k}{4}-(z-w)}$. Putting all the pieces together, the number of strings that correspond to distances $\delta_1$ and $\delta_2$ is upper-bounded by:

\begin{align*}
  \sum_{w = 0}^{\min \{z+(\delta_1-\delta_2)/2, \; \delta_1/2 - k/4\}} A_w B_w C_w D_w
\end{align*}

To upper bound $v |\Q|$ we sum this expression over all the possible distances $\delta_1 \le \delta_2$, and then multiply by two (by symmetry). 

Now note that the terms $A_w, B_w, D_w$ are maximized for $\delta_1=\delta_2=k$, for any values of $z$ and $w$. 
This is clear for $D_w$. 
For $A_w$ and $B_w$ recall that $\dbinom{a}{b}$ is maximized for $b=a/2$, and for $b$ the closest possible to $a/2$ if there are some constraints preventing $b=a/2$. 
Thus, $A_w$ is maximized for $\delta_1=\delta_2$ and $B_w$ is maximized for $\delta_1=k$.
Furthermore, the upper bound on $w$ is maximised for $\delta_1=\delta_2=k$. We therefore obtain that $v |\Q|$ is upper-bounded by

\begin{equation}\label{eq:v-first-bound}
\sum_{w = 0}^{\min \{z, k/4\}} k^2  {\dbinom{2z}{z}} {\dbinom{k}{k/4-w}}  {\dbinom{k}{k/2}} (d/k)^{3k/4-(z-w)}
\end{equation}

Next, we upper-bound this sum through two steps, Claim~\ref{claim:w} and Claim~\ref{claim:z}.  

\begin{claim}\label{claim:w}
For any $z$, the maximum term of the sum in Equation~\ref{eq:v-first-bound} is for the value ${w=\min\{z,k/4\}}$. 
\end{claim}
\begin{proof}
For all $w < k/4$, and $d$ large enough, the ratio between consecutive terms is equal to 
\[ 
\frac{(d/k) (k/4-w)}{(w+1+3k/4)} > 1
\]
thus the largest term of the sum is for the largest $w$ in the sum, that is ${w=\min\{z,k/4\}}$.
\end{proof}
As a consequence of Claim~\ref{claim:w}, we get that  

\begin{equation}\label{eq:v-second-bound}
v |\Q| \le k^2 \times \frac{k}{4} {\dbinom{2z}{z}} {\dbinom{k}{k/4-\min\{z,k/4\}}}{\dbinom{k}{k/2}} (d/k)^{3k/4- (z-\min\{z,k/4\})}
\end{equation}
And we simplify this upper bound again, by studying the influence of~$z$.

\begin{claim}\label{claim:z}
The maximum value of the right-hand side of Equation~\ref{eq:v-second-bound} is reached for $z=k/16$.
\end{claim}

\begin{proof}
We consider two cases, $k/16 \le z \le k/4$ and $z \ge k/4$. Consider first $k/16 \le z \le k/4$. In this case, Equation~\ref{eq:v-second-bound} becomes: 

\begin{equation}\label{eq:z<k/4}
v |\Q| \le k^2 \times \frac{k}{4} {\dbinom{2z}{z}} {\dbinom{k}{k/4-z}}  {\dbinom{k}{k/2}} (d/k)^{3k/4}
\end{equation}
Once again we compute the ratio of the expression for consecutive values of $z$:

\[
\frac{\dbinom{2(z+1)}{z+1} \dbinom{k}{k/4-(z+1)}}{\dbinom{2z}{z}\dbinom{k}{k/4-z}} = \frac{(2z+2)(2z+1)}{(z+1)^2} \cdot \frac{(k/4-z)}{(3k/4+z+1)}
\]
Let us prove that for $z \ge k/16$, this expression is smaller than $1$. 
The first ratio tends toward~$4$. The second ratio is decreasing as a function of~$z$, thus it is maximized for the smallest value of~$z$ possible, which is~$z=k/16$. For this value, this second term is upper bounded by $3/13$. Thus the ratio is asymptotically bounded by $12/13$ which is strictly smaller than $1$. 

Therefore, the maximum of the right-hand term of Equation~\ref{eq:z<k/4} is reached for the smallest value of $z$ which is $z = k/16$, and the claim follows. Suppose now $ z \geq k/4$. From above we have

\[
v |\Q|  \le k^2 \times\frac{k}{4} {\dbinom{2z}{z}}  {\dbinom{k}{k/2}} (d/k)^{k-z}
\]
The ratio of the expression for consecutive values of $z$ is $\frac{(2z+1)(2z+2)}{(z+1)^2} \cdot (k/d)$. Recall that $n \le d (d/k)^k$ and therefore $d \ge k^{k/(k+1)} n^{1/(k+1)}$. Therefore, for $n$ large enough, the ratio is smaller than $1$, hence the maximum is reached for $z = k/4$ in this case. The same conclusion follows. 
\end{proof}
For $z=k/16$, Equation~\ref{eq:v-second-bound} is:
\[
v |\Q| 
\le 
k^2 \times
\frac{k}{4} {\dbinom{k/8}{k/16}} {\dbinom{k}{3k/16}} {\dbinom{k}{k/2}} (d/k)^{3k/4}
\]
and this finishes the proof of Lemma~\ref{lem:fraction-Q}.
\end{proof}

\subsection{Wrapping up.}
In the previous subsection we defined regions and points, and bounded the key parameters. 
We can now apply Theorem~\ref{th:pointer_lower_bound_gen}. It tells us that if the query time of a data structure for dictionary look-up for Hamming distance $k$ is $g(n) \le t$, then it must use space $\Omega(t / v)$. 
Assume that the hypothesis on the query time is true. We now lower bound $t/v$.

\begin{claim}
There exists $c>2$ such that $t / v = \Omega \left(n \cdot c^{k/16}\right)$.
\end{claim}

\begin{proof}
Substituting $t$ and $v$ by their estimates (from Lemma~\ref{lem:regions}, Fact~\ref{fact:Q} and Lemma~\ref{lem:fraction-Q}), we obtain
\begin{align*}
  t/v = \Omega\left(\frac{p_\alpha {\dbinom{k}{k/4}} 2^{k/4} (d/k-2)^{k/4} {\dbinom{k}{k/2}} (d/k)^k (d/k-1)^{k/2}}  
  {k^2\times \frac{k}{4} {\dbinom{k/8}{k/16}} {\dbinom{k}{3k/16}} {\dbinom{k}{k/2}} (d/k)^{3k/4}}\right)
\end{align*}
Remember that $n = |\D| \cdot d = \Theta(p_\alpha d \cdot (d/k)^{k})$. Hence, $p_\alpha(d/k)^k=\Theta(n/d)$, and the above expression leads to:

\begin{equation*}
t / v =  \Omega\left(\frac{n}{k^3} \times \frac{ {\dbinom{k}{k/4}} 2^{k/4}} {{\dbinom{k/8}{k/16}} {\dbinom{k}{3k/16}}}  
\times \frac{(d/k-1)^{k/2} (d/k-2)^{k/4}}{ d (d/k)^{3k/4}}\right)
\end{equation*}
We estimate the second term. First note that $2^{k/8} \ge  {\dbinom{k/8}{k/16}}$. We also have 
${\dbinom{k}{k/4}}\geq {\dbinom{k}{3k/16}}$
which implies that the term is at least $2^{k/8}$. 
Being more careful with the ratio actually gives a lower bound of $c_1^{k/8}$ for $c_1>2$. 
Therefore, we get the following lower bound.

\begin{equation}\label{eq:t/v-bound}
t / v = \Omega \left( n \cdot \frac{c_1^{k/8}}{k^3} \cdot \frac{(d/k-1)^{k/2} (d/k-2)^{k/4}}{d (d/k)^{3k/4}} \right)
\end{equation}
We need more bounds on $d$. To get those let us first lower bound $p_{\alpha}$:
\[
p_\alpha = 
\frac{1}
{\sum_{i = 0}^{\lfloor{k \left(\frac{1}{4} - \alpha\right)\rfloor}} 
\dbinom{d}{i}}
\geq 
\frac{1}
{k/4\cdot (16d/k)^{k/4}}
\]
From 
$n = \Theta \left(p_\alpha d (d/k)^{k}\right)$,  and 
$p_\alpha \ge \frac{1}{k/4\cdot (16 d/k)^{k/4}}$, 
we obtain (after simplifications):

\[
d = \Omega (k^{k/(k+1)} n^{1/(k+1)})
\ \text{and}\ 
d = \Oh\left(k n^{1/(3k/4 + 1)}\right)
\]
Let us go back to the right-hand term of Equation~\ref{eq:t/v-bound}. We isolate the term $n/d$, and consider the rest of the expression.
We can choose $n$ large enough so that for some $c_2 > 2$: 

\[
c_1^{k/8} (d/k-1)^{k/2} (d/k-2)^{k/4} / (k^3 (d/k)^{3k/4}) \ge c_2^{k/8}
\]  
Finally, $d=\Oh(k n^{1/(3k/4 + 1)})$, which implies $d = \Oh(k2^{k/16})$ for 
$\frac{8}{\sqrt{3}} \sqrt{\log n} \le k = o(\log n)$. 
Therefore, $t / v = \Omega \left(n \cdot c_2^{k/16}\right)$.
\end{proof}

We now consider the condition $g(n) \le t$. 
Here $g(n) = d + g'(n)$, where $d$ is the time we spend for reading a query string from $Q$ and $g'(n)$ is the function we want to estimate. 
Recall from above that 

$$t = \Omega\left(p_\alpha 2^{k/4} {\dbinom{k}{k/4}} (d/k-2)^{k/4}\right) \mbox{ and } p_\alpha \ge 1/[(k/4) (16d/k)^{k(1/4-\alpha)}]$$
By Stirling's approximation,

$${\dbinom{k}{k/4}} = \Theta \left(\frac{k^k}{\sqrt k (k/4)^{k/4} (3k/4)^{3k/4}}\right) = \Omega (4^{k/4} (4/3)^{3k/4} / \sqrt k) = \Omega (4^{k/4} (2.37)^{k/4} / \sqrt k)$$
Hence we obtain that 

$$2^{k/4} {\dbinom{k}{k/4}} (d/k-2)^{k/4} = \Omega(8^{k/4} (d/k - 2)^{k/4} (2.37)^{k/4} / k \sqrt k)$$
Plugging this into the expression for $t$ and paying $(2.37/2)^{k/4}$ for replacing $d/k-2$ with $d/k$ and removing $1/ k^2 \sqrt k$, we obtain

$$t = \Omega (p_\alpha 2^{k/4} {\dbinom{k}{k/4}} (d/k - 2)^{k/4}) = \Omega ((d/k)^{k \alpha}) = \Omega ((\log n/k)^k)$$
Therefore, we can take $g'(n) = \Oh ((\log n/k)^k)$. We now want to show that $d \le t$. 
Recall that $d = \Oh\left(k n^{1/(3k/4 + 1)}\right)$. We use the fact that $k=o(\log n)$, and rewrite this expression as a power of $\log n$: 
\[
d=\Oh\left((\log n)^{1+\frac{\log n}{\log\log n (3k/4 + 1)}}\right)
\]
As $\frac{8}{\sqrt{3}} \sqrt{\log n}\leq k$, the exponent is in $\Oh\left(\frac{\sqrt{\log n}}{\log \log n}\right)$. Now let us lower bound $t$ further. To do so, as in the proof of Claim~\ref{claim:alpha-zero}, let $h=\log n/k$. Using the lower bound on $k$ for the exponent, we get $t = h^{\Omega(\sqrt{\log n})}$.
In order to compare $t$ and $d$, we take the logarithm for both expressions, and after simplifications:

\[
\log d = \Oh\left(\sqrt{\log n}\right)
\text{ and }
\log t= \Omega\left(\sqrt{\log n}\cdot \log h\right) 
\] 
As $h$ tends toward infinity (because $k=o(\log n)$), we finally get that asymptotically $d \le t$. Theorem~\ref{th:pointer_lower_HD_dict} follows, that is, to have query time $\Oh(d + (\frac{\log n}{k})^k + occ)$ we must use $\Omega(n \cdot c^{k/16})$ space, for some constant $c > 2$.

\begin{remark}
One may wonder if the lower bound can be improved by tweaking the construction a little bit. We argue that this is not the case: The distance between the dictionary strings must be at most $k/4$ so that the number of the regions in Claim~\ref{claim:LB-regions} gives a meaningful bound for the query time. On the other hand, the larger the distance is, the smaller the intersection of the regions is, and we choose the distance to be as close to $k/4$ as possible.
\end{remark}

\subsection{Dictionary look-up with differences.}
We finish this section with the following corollary.

\begin{corollary}\label{cor:pointer-machine-dictionary-ED}
Assume the pointer-machine model. For all even $\frac{8}{\sqrt{3}} \sqrt{\log n} \le k = o(\log n)$ there exists $d$ such that any data structure for dictionary look-up with $k$ differences on a set of $n$ strings of length $d$ that has query time $\Oh(d / \log d + \left(\frac{\log n}{2k}\right)^k + occ)$ requires $\Omega(c^k n)$ space, for some constant~$c > 1$.  
\end{corollary}
\begin{proof}
We use the stoppers transform (see Section~\ref{sec:stoppers}) to construct a hard instance of dictionary look-up with $k$ differences. Recall that the edit distance between any two strings in the differences instance is equal to the Hamming distance between the preimages of these strings in the original instance. We can therefore use dictionary look-up with differences  to solve the instance of dictionary look-up with mismatches. If $d'$ is the length of an individual string and $n'$ is the total length of the strings in the instance of dictionary look-up with mismatches, then $d \le c_1 \cdot d' \log d'$ and $n \le c_1 n' \log d \le c_1 n' \log n$. It follows, in particular, that $\log n \le 2 \log n'$. Recall that the space lower bound for mismatches is $\Omega((c')^{k/16} n)$ for some constant $c' >2$. Since $k \ge \frac{8}{\sqrt{3}} \sqrt{\log n}$, we have $(c')^{k/16} n' \ge c^{k} c_1 n' \log n' \ge c^{k} n$ for some constant $c > 1$. The claim follows.
\end{proof}

As a remark, the size of the alphabet of the dictionary strings in Corollary~\ref{cor:pointer-machine-dictionary-ED} is $\Oh(\log(k n^{1/(3k/4+1)})) = \Oh(\sqrt{\log n})$.

\section{Lower bounds for text indexing}\label{sec:textindexing}
In this section we show time lower bounds for text indexing with mismatches and differences by reduction from dictionary look-up with mismatches and differences. In order to do that, we use a very simple idea: suppose towards contradiction that there is an efficient data structure for text indexing. We construct a hard instance of text indexing from a hard instance of dictionary look-up. Namely, to obtain the text we concatenate the dictionary strings together interleaving them with a special gap string $G$ that we will define later to obtain the text. We can then solve the dictionary look-up problem for a string $Q$ by searching for $GQG$ in the text, using the efficient data structure for text indexing and so beating the lower bound for dictionary look-up. The choice of the gap string $G$ must guarantee that the text index always returns a substring of the text that corresponds to a dictionary string as an answer. For the case of mismatches, we show existence of $G$ using the probabilistic method. For the case of differences, we provide an explicit construction.

\subsection{Text indexing with mismatches.}
\begin{theorem}\label{th:dictionary-textindexing-HAM}
Let $0 < \eps < 1$ and $(1+\eps) k < d$. 
If there is a data structure for exact (resp., $(1+\eps)$-approximate) text indexing with $k$ mismatches that for a text of length $n$ has construction time $T_{c} (n)$, query time $T_q (n)$, and space $S(n)$, then there is a data structure for exact (resp., $(1+\eps)$-approximate) dictionary look-up with $k$ mismatches that for a dictionary of $d$-length strings of total length $n$ has construction time $T_{c} (3n + 2d) + \mathrm{poly} (d^{\log d})$, query time $T_q (3n + 2d)$, and space $S(3n + 2d)$. 
\end{theorem}
\begin{proof}
Suppose that the dictionary contains strings $S_1, S_2, \ldots, S_{n/d}$, where each string has length~$d$. 
To obtain a text $T$ of length $3n+2d$, we concatenate them together interleaving with a gap string $G$, i.e. $T = G S_1 G S_2 G \ldots S_{n/d} G$. 
The gap string $G \in \{\$,\#\}^{2d}$, where $\$,\#$ are two symbols that do not belong to the alphabet of $S_1, S_2, \ldots, S_{n/d}$ and the query strings. We have the following claim:

\begin{claim}
\label{cl:gapstring}
There exists a gap string $G$, such that for any $3d/2 \le i < 2d$, the Hamming distance between the prefix and the suffix of $G$ of length $i$ is at least $d / 2 + 1$. We can find $G$ in $d^{\Oh(\log d)}$ time.
\end{claim}
  
Let us defer the proof of the claim for the moment. It is enough to prove that for any string $Q$ of length $d$, the text index can output only substrings of $T$ of form $G S_i G$. This is because if the Hamming distance between $GQG$ and $GS_iG$ is at most $k$ (or $(1+\eps) k$ for the approximate case), the Hamming distance between $Q$ and $S_i$ is at most $k$ (resp., $(1+\eps) k$) as well. 
 
Suppose this is not the case. By definition, if the output of the dictionary look-up problem is not empty, the text index will output a substring of length $|GQG| = 5d$ . Let $0 < \Delta < 3d$ be the distance from the starting position of the output substring to the starting position of the nearest copy of $G$ from the right. Suppose first that $d/2 < \Delta < 5d/2$. In this case, the total length of the overlaps of the first copy of $G$ in $GQG$  with a dictionary string $S_i$ and $Q$ with a copy of $G$ in $T$ is at least $d$. Recall that the alphabet of $G$ is different from the alphabet of the dictionary strings $S_i$ and the query string $Q$, meaning that the Hamming distance between $GQG$ and the output substring is at least $d > (1+\eps) k$ in this case, a contradiction. Suppose now that $\Delta \in (0, d/2] \cup [5d/2,3d)$. In this case, each of the two copies of $G$ in $Q$ overlaps some copy of $G$ in $T$ by a prefix of length $\ge 3d / 2$ (or a suffix of length $\ge 3d/2$, which is symmetrical), and the Hamming distance between them is at least $d / 2 + 1$. Therefore, the total Hamming distance is $> d$ again, a contradiction. 

It follows that the text index can only output substrings of $T$ where $\Delta = 0$, i.e. substrings of form $G S_i G$.
\end{proof}

\bigskip

\noindent \textit{Proof of Claim~\ref{cl:gapstring}.}
We first show that such a gap string $G$ exists by the probabilistic method. 
Consider a string of length $2d$ over the alphabet $\{\$,\#\}$, where the symbols are $(2 \lceil{\log d\rceil})$-wise independent and uniformly distributed. Consider an integer $i$, $3d/2 \le i < 2d$. We will show that with high probability the Hamming distance $\Ham(G[1,i], G[2d-i+1,2d])$ is larger than $d/2+1$. 

In order to do that, we divide both strings into non-overlapping blocks of length $\lceil \log d \rceil$ (the last blocks may be shorter). We then show that the Hamming distance between any $\lceil \log d \rceil$-length block of $G[1,i]$ and the corresponding block of $G[2d-i+1,2d]$ is at least $5 \lceil \log d \rceil / 12$ with high probability. We define a family of random variables $\chi_j$, where $\chi_j$ is equal to~$1$ if there is a mismatch at the position $j$ of the blocks, and to~$0$ otherwise. Since the symbols of $G$ are $(2 \lceil{\log d\rceil})$-wise independent and uniformly distributed, the variables~$\chi_j$ are independent and uniformly distributed as well. 
We now apply the Chernoff--Hoeffding bounds~\cite[Theorem 1]{MR0144363}.  For $\lambda = \log d$ independently and identically distributed binary variables $\chi_1, \chi_2, \ldots, \chi_\lambda$, 

\[
\prob {\frac{1}{\lambda}\sum_{i=1}^\lambda \chi_j \le \mu-\gamma} 
\le e^{-2\lambda \gamma^2}
\]
for every $0<\gamma <\mu $, where $\mu = \prob {\chi_j = 1} = 1/2$. Therefore, the probability that the two blocks have Hamming distance smaller or equal to $5 \lceil \log d \rceil / 12$ is at most $e^{-2 [ (\log d) / 12 ]^2}$. 

By applying the union bound over the blocks, we obtain that the Hamming distance $\le \lfloor{ \frac{i}{\lceil{ \log d \rceil}} \rfloor} \cdot (5 \lceil \log d \rceil / 12)$ (note that the right-hand side of this inequality is $\ge d/2 + 1$ for $i \ge 3d/2$) with probability at most $\lfloor {\frac{i}{\lceil \log d \rceil}\rfloor} e ^{-2 [ (\log d) / 12 ]^2]}$ (note that the probability reaches its maximum for $i = 2d-1$). By applying another union bound over all $i$, $3d/2 \le i < 2d$, we obtain that the probability of the Hamming distance to be smaller than $d/2+1$ for some $i$ is at most ${\frac{4 d^2}{\log d} e^{-[ (\log d) / 12 ]^2]} < 1}$. Therefore, the string $G$ that we search for exists. 

However, we also need to be able to construct $G$ efficiently. We use a standard construction of $d$ uniform $2 \lceil{\log d\rceil}$-independent binary random variables, which requires $\Oh(\log^2 d)$ random bits. (See, for example,~\cite[Section 4.1]{LectureHastad}). We can iterate over all possible sets of random bits in $d^{\Oh(\log d)}$ time, and for each set we can check if the corresponding string has the desired property in $\Oh(d^2)$ time, which concludes the proof.

\bigskip

From Theorem~\ref{th:dictionary-textindexing-HAM} and Corollary~\ref{cor:hardness_DL_approx} we immediately obtain a lower bound for text indexing with mismatches conditional on SETH. In this version of text indexing, we assume that we output only one substring of the text. It is important that the construction time of the data structure for dictionary look-up that we obtain by the reduction remains polynomial, as $d = \Theta(\log n)$.

\begin{corollary}\label{cor:SETH-text-indexing-HAM}
Assuming SETH, for any constant $\delta > 0$ there exists $d = \Theta(\log n)$, and $k = \Theta(\log n)$ such that 
any data structure for text indexing with $k$ mismatches for a string of length $n$ that can be constructed in polynomial time has query time $\Omega (n^{1-\delta})$.
\end{corollary}

\begin{corollary}\label{cor:SETH-approx-text-indexing-HAM}
Assuming SETH, for any constant $\delta > 0$ there exists $\eps = \eps(\delta)$, $d = \Theta(\log n)$, and $k = \Theta(\log n)$ such that any data structure for  $(1+\eps)$-approximate text indexing with $k$ mismatches for a string of length $n$ that can be constructed in polynomial time has query time $\Omega (n^{1-\delta})$. 
\end{corollary}

The next claim follows from Theorems~\ref{th:pointer_lower_HD_dict} and~\ref{th:dictionary-textindexing-HAM}. Here we assume that the text index outputs all substrings of the text that are close to the query. Note that we decrease the second term in the time complexity to account for the extra factor in the total length of the text that appears in Theorem~\ref{th:dictionary-textindexing-HAM}.

\begin{corollary}\label{cor:pointer-machine-text-indexing-HAM}
For all even $k$, $\frac{8}{\sqrt 3} \sqrt{\log n} \le k = o(\log n)$, if a pointer-machine data structure for text indexing with $k$ mismatches for a text of length $n$ has query time $\Oh(d + (\frac{\log n}{2k})^k + occ)$, it must use $\Omega(c^k n)$ space, for some constant $c > 1$. 
\end{corollary}

\subsection{Text indexing with differences.}
We now show a reduction from dictionary look-up with differences to text indexing. 

\begin{theorem}\label{th:dictionary-textindexing-ED}
Let $0 < \eps < 1$ and $(1+\eps) k < d$. 
Suppose there is a data structure for exact (resp., $(1+\eps)$-approximate) text indexing with $k$ differences that for a text of length $n$ has construction time $T_{c} (n)$, query time $T_q (n)$, and space $S(n)$. Then there is a data structure for exact (resp., $(1+\eps)$-approximate) dictionary look-up with $k$ differences that for a dictionary of total length $n$ has construction time $T_{c} (3n+2d)$, query time $T_q (3n + 2d)$, and space $S(3n+2d)$. 
\end{theorem} 
\begin{proof}
Suppose that the dictionary contains strings $S_1, S_2, \ldots, S_{n/d}$, where each string has length~$d$. To obtain a text $T$ of length $3n+2d$, we concatenate them together interleaving with a gap string $G$, i.e. $T = G S_1 G S_2 G \ldots S_{n/d} G$. The gap string $G$ is defined to be equal to $\underbrace{\$\$ \dots \$}_{d}\underbrace{\#\# \dots \#}_{d}$
where $\$,\#$ are two symbols that do not belong to the main alphabet. Furthermore, for a query string $Q$ in the instance of the dictionary look-up problem, we define the pattern in the instance of text indexing to be $G Q G$. 

For any dictionary string $S_j$, the edit distance between $GS_j G$ and $GQG$ is at most the edit distance between $S_j$ and $Q$, and therefore if the output of the dictionary look-up problem is not empty, the output of the text index is not empty as well. Suppose that the text index outputs a substring $T'$ of $T$. Both in the case of the exact and the approximate text indexes the edit distance between $GQG$ and $T'$ must be at most $(1+\eps) k < d$. Consider an optimal alignment of the pattern $GQG$ and $T'$. Suppose that the left copy of $G$ in the pattern is aligned with a prefix $L$ of $T'$, $Q$ with a substring $M$ of $T'$, and the right copy of $G$ with a suffix $R$ of $T'$. First note that neither~$L$, nor $R$ can contain an entire copy of a dictionary string, since the alphabet of dictionary string is different from the alphabet of $G$ and if this were the case the edit distance between $G$ and $L$ (or $R$) would be $\ge d$. Analogously, $M$ cannot contain an entire copy of $G$. Therefore, $M$ can be represented as $T[(3j+2)d+1+\Delta_1,(3j+3)d+\Delta_2]$, where $S = T[(3j+2)d+1,(3j+3)d]$ is a dictionary string, and $-2d < \Delta_1 \le d$, $-3d \le \Delta_2 \le 3d$. We will show that the edit distance between $Q$ and $S$ is at most the edit distance between $GQG$ and $T'$, and therefore $S$ can be output as an answer to the dictionary look-up query. Let us start by giving more precise bounds for $\Delta_1, \Delta_2$.

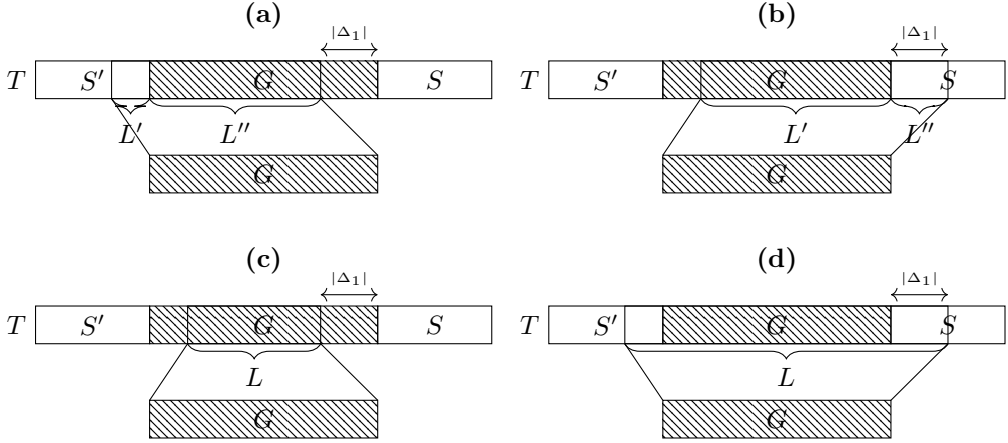
\begin{figure*}[ht!]
\begin{center}
\begin{tikzpicture}[scale=0.5]

\begin{scope}
\draw[pattern=north west lines] (0,2) rectangle (6,3) node[pos=0.5] {$G$};
\draw (-3,2) rectangle (0,3) node[pos=0.5] {$S'$};
\draw (6,2) rectangle (9,3) node[pos=0.5] {$S$};
\draw[pattern=north west lines] (0,-0.5) rectangle (6,0.5)  node[pos=0.5] {$G$};
\draw (0,0.5)--(-1,2);
\draw (6,0.5)--(4.5,2);

\draw [decorate, decoration={brace, mirror, amplitude=2mm}] (-1,2) -- (0,2) node [midway, below=2mm] {$L'$};
\draw [decorate, decoration={brace, mirror, amplitude=2mm}] (0,2) -- (4.5,2) node [midway, below=2mm] {$L''$};

\draw (-1,2) rectangle (4.5,3) ;
\draw[<->] (4.5,3.3)--(6,3.3) node[midway,above] {\tiny{$|\Delta_1|$}};
\node at (3,4.2) {\textbf{(a)}};
\node[left] at (-3,2.5) {$T$};
\end{scope}

\begin{scope}[xshift=13.5cm]
\draw[pattern=north west lines] (0,2) rectangle (6,3) node[pos=0.5] {$G$};
\draw (-3,2) rectangle (0,3) node[pos=0.5] {$S'$};
\draw (6,2) rectangle (9,3) node[pos=0.5] {$S$};
\draw[pattern=north west lines] (0,-0.5) rectangle (6,0.5)  node[pos=0.5] {$G$};
\draw (0,0.5)--(1,2);
\draw (6,0.5)--(7.5,2);

\draw [decorate, decoration={brace, mirror, amplitude=2mm}] (1,2) -- (6,2) node [midway, below=2mm] {$L'$};
\draw [decorate, decoration={brace, mirror, amplitude=2mm}] (6,2) -- (7.5,2) node [midway, below=2mm] {$L''$};

\draw (1,2) rectangle (7.5,3) ;
\draw[<->] (6,3.3)--(7.5,3.3) node[midway,above] {\tiny{$|\Delta_1|$}};
\node at (3,4.2) {\textbf{(b)}};
\node[left] at (-3,2.5) {$T$};
\end{scope}

\begin{scope}[yshift=-6.5cm]
\draw[pattern=north west lines] (0,2) rectangle (6,3) node[pos=0.5] {$G$};
\draw (-3,2) rectangle (0,3) node[pos=0.5] {$S'$};
\draw (6,2) rectangle (9,3) node[pos=0.5] {$S$};
\draw[pattern=north west lines] (0,-0.5) rectangle (6,0.5)  node[pos=0.5] {$G$};
\draw (0,0.5)--(1,2);
\draw (6,0.5)--(4.5,2);

\draw [decorate, decoration={brace, mirror, amplitude=2mm}] (1,2) -- (4.5,2) node [midway, below=2mm] {$L$};

\draw (1,2) rectangle (4.5,3) ;
\draw[<->] (4.5,3.3)--(6,3.3) node[midway,above] {\tiny{$|\Delta_1|$}};
\node at (3,4.2) {\textbf{(c)}};
\node[left] at (-3,2.5) {$T$};
\end{scope}

\begin{scope}[yshift=-6.5cm, xshift = 13.5cm]
\draw[pattern=north west lines] (0,2) rectangle (6,3) node[pos=0.5] {$G$};
\draw (-3,2) rectangle (0,3) node[pos=0.5] {$S'$};
\draw (6,2) rectangle (9,3) node[pos=0.5] {$S$};
\draw[pattern=north west lines] (0,-0.5) rectangle (6,0.5)  node[pos=0.5] {$G$};
\draw (0,0.5)--(-1,2);
\draw (6,0.5)--(7.5,2);

\draw [decorate, decoration={brace, mirror, amplitude=2mm}] (-1,2) -- (7.5,2) node [midway, below=2mm] {$L$};

\draw (-1,2) rectangle (7.5,3) ;
\draw[<->] (7.5,3.3)--(6,3.3) node[midway,above] {\tiny{$|\Delta_1|$}};
\node at (3,4.2) {\textbf{(d)}};
\node[left] at (-3,2.5) {$T$};
\end{scope}

\end{tikzpicture}
\end{center}
\caption{Four possible locations of $L$ in $T$, here $S,S'$ are dictionary strings.}
\label{fig:LR}
\end{figure*}

\begin{fact}\label{fact:boundaries_Delta}
We have $-d \le \Delta_1, \Delta_2 \le d$.
\end{fact}
\begin{proof}
Let us first show the claim for $\Delta_1$. If $-2d \le \Delta_1 \le -d$, then $\Delta_2 \le - d$ as well, because otherwise $M$ contains at least $d$ symbols in $\{\$,\#\}$ and the edit distance between $Q$ and $M$ is at least $d$, which implies that the edit distance between $GQ G$ and $T'=LMR$ is at least $d$, a contradiction. Recall that neither $L$, no $R$ can contain a dictionary string. It follows that the total length of $L, M, R$ is at most $4d$, and therefore the edit distance between $GQG$ and $T' = LMR$ is at least $d$, a contradiction. It follows that $-d \le \Delta_1 \le d$.

We now show the claim for $\Delta_2$. Suppose first that $\Delta_2 \le -d$. In this case, similar to above, we can show that $|T'| \le 4d$, a contradiction. If $\Delta_2 \ge d$, then from $\Delta_1 \le d$ it follows that $M$ contain at least $d$ symbols in $\{\$,\#\}$ and the edit distance between $Q$ and $M$ is at least $d$, a contradiction.
\end{proof}

\begin{lemma}
Let $ED_L$ (resp., $ED_R$) be the edit distance between the left copy of $G$ and $L$ (resp., $R$), then $ED_L \ge |\Delta_1|$ (resp. $ED_R \ge |\Delta_2|$). 
\end{lemma}
\begin{proof}
We give the proof for $ED_L$, for $ED_R$ it follows from symmetry. 
Fact~\ref{fact:boundaries_Delta} implies that we must consider four cases for $L$, see Fig.~\ref{fig:LR}.

\textbf{Case (a).} Let $L = L' L''$, where $L'$ is a suffix of $S'$, and $L''$ is the prefix of $G$. Let $G', G''$ be the prefix and the suffix of the first copy of $G$ in $G Q G$ aligned with $L', L''$ respectively. The edit distance between $G'$, $L'$ is equal to $\max\{|G'|, |L'|\}$, because they have different alphabets. If $|G'| > d$, we are done. We suppose now that $|G'| \le d$ and estimate the edit distance between $G''$ and $L''$. If we delete at least $|\Delta_1|$ symbols from these two strings, then the edit distance is at least $|\Delta_1|$, and we are done. Suppose now we delete $x < |\Delta_1|$ symbols. From the definition of the edit distance it follows that it equals to the total number of deleted symbols plus the Hamming distance between the remaining strings. After deletion $G''$ contains at least $d-x$ symbols \#, while $L''$ contains at most $d-|\Delta_1|$ symbols \#. Therefore, the Hamming distance between the two strings is at least $|\Delta_1|-x$, and consequently the edit distance between $G''$ and $L''$ is at least $|\Delta_1|$. Therefore, the edit distance between $G$ and $L$ is at least $|\Delta_1|$.

\textbf{Case (b).} We divide $L$ into two parts, $L'$ that is a suffix of $G$ and $L''$ that is a prefix of $S$ of length $|\Delta_1|$. Suppose that they are aligned with a prefix $G'$ and a suffix $G''$ of $G$, respectively. Since $L''$ and $G''$ have different alphabets, the edit distance between them is $\max\{|L''|, |G''|\} \ge |\Delta_1|$. 

\textbf{Cases (c) and (d).} We estimate the edit distance between $G$ and $L$ by the absolute value of the difference of their lengths, which is larger than $|\Delta_1|$. 
\end{proof}

As a corollary, the edit distance between $Q$ and $S$ is at most the edit distance between $GQG$ and $T'$, which means that we can return $S$ as the answer for the dictionary look-up query (both for exact and approximate data structures). Indeed, the edit distance between $M$ and $T[(3j+2)d+1,(3j+3)d] = S$ is at most $|\Delta_1|+|\Delta_2|$. Therefore, by the triangle inequality, 

\begin{align*}
\ED(Q,S) &\le \ED(Q,M) + \ED(M,S) \le \ED(Q,M) + |\Delta_1|+|\Delta_2| \le \\
& \le \ED(Q,M) + \ED(G,L) + \ED(G,R)
\end{align*}
Since the right-hand side of this inequality equals to the edit distance between $GQG$ and $T'$ by definition, the claim follows.
\end{proof}

From Theorem~\ref{th:dictionary-textindexing-ED} and Corollaries~\ref{cor:hardness_DL_approx},~\ref{cor:DLD} we immediately obtain conditional lower bounds for text indexing with differences. Again, in this version of text indexing, we assume that we output only one substring of the text. It is important that the construction time of the data structure for dictionary look-up that we obtain by the reduction remains polynomial, as $d = \Theta(\log n \log \log n)$.

\begin{corollary}\label{cor:SETH-text-indexing-ED}
Assuming SETH, $\forall \delta > 0$ $\exists d = \Theta(\log n)$, and $k = \Theta(\log n)$ such that any data structure for text indexing with $k$ differences for a string of length $n$ that can be constructed in polynomial time has query time $\Omega (n^{1-\delta})$.
\end{corollary}

\begin{corollary}\label{cor:SETH-approx-text-indexing-ED}
Assuming SETH, $\forall \delta > 0$ $\exists \eps = \eps(\delta)$, $d = \Theta(\log n)$, and $k = \Theta(\log n)$ such that any data structure for $(1+\eps)$-approximate text indexing with $k$ differences for a string of length $n$ that can be constructed in polynomial time has query time~$\Omega (n^{1-\delta})$.
\end{corollary}

\bibliographystyle{abbrv}
\bibliography{blast.bib}
\end{document}